\documentclass[12pt,a4paper]{article}
\usepackage{bbm}
\usepackage{amssymb,amsmath}
\usepackage{rotating}
\usepackage[centerlast]{caption}
\usepackage{graphicx}
\usepackage{verbatim}
\usepackage{enumerate}
\usepackage{mathrsfs}
\usepackage{amsfonts}
\usepackage{multirow}
\usepackage{amsthm,amscd,graphics,latexsym,color}
\usepackage{booktabs,tabularx}
\usepackage{caption}
\usepackage{subcaption}
\usepackage[title,titletoc]{appendix}
\usepackage[authoryear]{natbib}
\usepackage{amsmath}
\usepackage{amssymb}

\usepackage[a4paper,text={17.8cm,26.4cm},centering]{geometry}

\def\be{\begin{equation}}
\def\ee{\end{equation}}
\def\bestar{\begin{eqnarray*}}
\def\eestar{\end{eqnarray*}}
\def\bea{\begin{eqnarray}}
\def\eea{\end{eqnarray}}
\theoremstyle{remark}

\theoremstyle{definition}
\newtheorem{defn}{Definition}
\theoremstyle{plain}
\newtheorem{theo}{Theorem}
\theoremstyle{plain}

\theoremstyle{plain}
\newtheorem{cor}{Corollary}
\theoremstyle{plain}

\theoremstyle{plain}
\newtheorem{lemma}{Lemma}

\numberwithin{equation}{section} \numberwithin{theo}{section}
\numberwithin{defn}{section} \numberwithin{rem}{section}
\numberwithin{cor}{section} \numberwithin{lemma}{section}
\numberwithin{prop}{section} \numberwithin{assumption}{section}
\allowdisplaybreaks[4]

\newcommand{\e}{{\mathbb E}}

\newcommand{\ra}{{\cal F}}

\begin{document}

\title{{Multiple--index Nonstationary Time Series Models}: \\{Robust Estimation Theory and Practice}}

{\footnotesize
\author{Chaohua Dong$^{a}$, Jiti Gao$^{b}$, Bin Peng$^{b}$ and Yundong Tu$^{c}$\footnote{Corresponding author. Address: Guanghua
School of Management and Center for Statistical Science, Peking University,
Beijing, 100871, China. E-mail: yundong.tu@gsm.pku.edu.cn.} \\[5mm]
$^{a}$Zhongnan University of Economics and Law, China\\
$^{b}$Monash University, Australia\\
$^{c}$Peking University, China
}}
\date{\today}
\maketitle

\begin{abstract}
This paper proposes a class of parametric multiple-index time series models that involve linear combinations of time trends, stationary variables and unit root processes as regressors. The inclusion of the three different types of time series, along with the use of a multiple-index structure for these variables to circumvent the curse of dimensionality, is due to both theoretical and practical considerations.
The M-type estimators (including OLS, LAD, Huber's estimator, quantile and expectile estimators, etc.) for the index vectors are proposed, and their asymptotic properties are established, with the aid of the generalized function approach to accommodate a wide class of loss functions that may not be necessarily differentiable at every point. 
The proposed multiple-index model is then applied to study the stock return predictability, which reveals strong nonlinear predictability under various loss measures. Monte Carlo simulations are also included to evaluate the finite-sample performance of the proposed estimators.
\medskip

{\bf Key words}: Additive index model, robust estimation; stationary process, time trend, unit root process
\medskip

{\bf JEL classification}: C13, C22

\end{abstract}

\newpage

\section{Introduction}

Robust model building and determination of good models for estimation and prediction is a very important part of much empirical research in economics and finance. Since real world data often display characteristics, such as nonlinearity, nonstationarity and spatiality, an important aspect of model building and determination is how to take such characteristics into account. Meanwhile, the increasing availability of economic and financial data in recent years has been accompanied by increasing interest in both theoretical research and empirical data analysis in diverse fields of sciences as well as more broadly within social and medical sciences.

One challenge is how to select the most relevant variables when there are many variables available from the data. One initial step is to explore a simple method associated with principal component analysis. Different from dealing with high-dimensional data issues in medical sciences for example, we do need to pay attention to many econometric issues, such as {multicollinearity} when aggregating real datasets for model building. Meanwhile, one additional complexity is that many original data appear to be highly dependent and even nonstationary. Simply using a transformed version of the original data may completely ignore some key features involved in the original data, such as trending behaviours, which can be a key interest in modelling time series data in climatology, energy, epidemiology, health and macroeconomics. More importantly, transformed versions of two different datasets may have similar features, but the original datasets can have very different structures.  Another challenge is the development of new models and estimation methods that are not only robust in theory, but also capable of offering user-friendly tools for both the computational and implementational procedures of  empirical data analysis in practice.

There are several levels of robustness involved in the discussion of this paper. First, the proposed model building and estimation methods are robust to data with outliers, spatiality and temporality, different types of nonstationarity. Second, the class of models we propose are applicable to several types of data, such as stationary time series data with deterministic trends, highly dependent and nonstationary time series data. Third, the proposed robust estimation procedures are iterative and computationally tractable. Fourth, the proposed models and robust estimation methods as well as the associated computational algorithms are user-friendly and easily implementable for empirical researchers and practitioners. While the relevant literature accounts for one or more of the model building and estimation issues we will address in this paper, the novelty of this paper is addressing them all simultaneously to produce new models and feasible estimation methods with computational tractability as well as empirical relevance and applicability.

We now propose a cointegrating relationship among time series sequence $\{y_t, x_t, z_t, \tau_t, 1\leq t \leq n\}$, where $y_t$ is a response variable, and we first specify the respective structures of $x_t$ and $z_t$ by
\begin{align}\label{m0}
\begin{split}
&x_t=x_{t-1}+w_t \ \ \ \mbox{and} \ \ \ z_t=h(\tau_t,v_t),
\end{split}
\end{align}
where $\tau_t=\frac{t}{n}$, as defined later, $w_t$ is a linear vector process, and $h(\cdot)$ is a known vector of functions driven by the time trend $\tau_t$ and stationary $v_t$.

In the formulation of \eqref{m0} $x_t$ is an $I(1)$ vector, and it is well known that $x_t$ may capture `large' trending behaviours. Meanwhile, $z_t$ is a function in ($\tau_t, v_t$) that can capture `small' fluctuations.  The general form of $z_t$ accommodates possible nonlinearity, as it is quite common in nonstationary time series settings. Such general structure covers some important special cases, such as (1) $z_t=h(\tau_t)+v_t$, a stationary vector superimposed with a vector of time trends; (2) $z_t=\sigma(\tau_t)^\top v_t$, a combination of $v_t$ with varying coefficients ($\sigma(\cdot)$ may be a matrix); and (3) $z_t = h(\tau_t) + \sigma(\tau_t)^\top v_t$, which accommodates both deterministic and stochastic trending behaviours in the respective mean and volatility components.

We therefore propose a parametric multiple-index model of the form:
\begin{align}
y_t=&g_1(x_{t}^\top\theta_0) + g_2(z_t^\top\beta_0)+e_t,\label{m1}
\end{align}
for $t=1,2,\cdots,n,$ where $\theta_0\in \mathbb{R}^{d_1}$ and $\beta_0\in \mathbb{R}^{d_2}$ are unknown parameters, both $g_1(\cdot)$ and $g_2(\cdot)$ are parametrically known functions, and $e_t$ is an error term. Different from existing settings, model (\ref{m1}) takes into account not only the `large' stochastic trending component driven by $x_t^{\top} \theta_0$, but also the `small' slowly varying component $z_t^{\top} \beta_0$ mainly represented by the deterministic time trend $\tau_t$. In doing so, one may consider $g_2(z_t^\top\beta_0)$ as a `conditional mean' decomposition of the `original' error term: $\xi_t = g_2(z_t^\top\beta_0) + e_t$ to eliminate any potential endogeneity involved in $e_t$. As a consequence, one may assume that $e_t$ is uncorrelated with $(x_t^{\top} \theta_0, z_t^{\top} \beta_0)$, although $\xi_t$ may be correlated with $x_t^{\top} \theta_0$. Our focus of this paper is that $x_t^{\top} \theta_0$ is nonstationary (rather than cointegrated), and the primary goal is to estimate these unknown parameters by M-estimation to take the advantage of its robustness for modelling nonlinear functional forms of nonstationary time series data.

To the best of our knowledge, model (\ref{m1}) has not been proposed and studied in the relevant literature. For the univariate regressor setting, \citet{chang2001} investigate a parametric additive model to accommodate a time trend, and both univariate stationary and unit root regressors, and the authors estimate the unknown parameters by a nonlinear least-square method. \citet{chang2003} study a parametric model that specifies the regression function as a mixture of linear function and nonlinear function of unit root process that appears as a single-index structure. Meanwhile, \citet{dgd2016} consider the estimation for a partially linear single-index model of I(1) process vector.
In addition, there are a couple of other models that may be considered as nonparametric counterparts of model (\ref{m1}). \citet{DL2018} and \citet{dlp2021} extend \citet{chang2001}, respectively, to additive nonparametric and fully nonparametric models involving time trend as well as stationary and univariate nonstationary regressors.

With respect to M-estimation, there is an extensive literature. Starting from the seminal work of \citet{huber1964}, M-estimation methods gradually extend to have general loss functions that may circumvent some drawbacks caused by outliers and so on. The resulting estimators yielded by such approaches are coined as $M$- and $Z$-estimators (see, for example, Chapter 3 of \citet{vaart1996}).

Using M-estimation, the unknown parameters $(\theta'_0, \beta'_0)'$ are estimated by minimizing
\begin{equation}\label{objfun}
L_n(\theta,\beta)=\sum_{t=1}^n\rho(y_t-g(x_{t}^\top\theta, z_t^\top\beta)),
\end{equation}
where $g(u, v) = g_1(u) + g_2(v)$ and $\rho(\cdot)$ is some loss function. Possible choices of $\rho(\cdot)$ include, but not limits to,
\begin{enumerate}[(1)]
\item Squared error loss: $\rho(u)=u^2$.
\item $L_p$ norm: $\rho(u)=|u|^p$ with $p>1$.
\item Huber's loss \citep{huber1964}: For fixed $c>0$, \begin{align*}
\rho_{c}(u)=\begin{cases}
\frac{1}{2}u^2, & |u|\le c,\\
c |u|-\frac{1}{2}c^2, & |u|>c.
\end{cases}
  \end{align*}
\item Least absolute deviation (LAD): $\rho(u)=|u|$.
\item Quantile loss \citep{koenker1978}: $\rho_{\tau}(u)=u(\tau-I(u<0))$ for $\tau\in (0,1)$.
\item Expectile loss \citep{newey1986}: $\rho_{\tau}(u)=|\tau-I(u<0)|u^2$ for $\tau\in (0,1)$.
\end{enumerate}

The loss function $\rho(\cdot)$ here may not be differentiable at every point on $\mathbb{R}$. Often in the literature the subgradient needs to be introduced to tackle these potential erratic points (\citealp{chenjia2010}). Alternatively, some papers consider the population criterion function like $\e[\ell(Z, \theta)]$ instead of $\ell(Z,\theta)$ itself when $Z$ is stationary variable. The rationale why the population criterion is workable is that although $\ell(Z,\theta)$ may not be differentiable at some points, $\e[\ell(Z, \theta)]$ may be smooth at these points (see, for example, \citet[p. 640]{chenxh2014}).

However, in this paper the model has nonstationary regressors that make the population criterion approach fail to work; also the so-called ``subgradient approach'' depends on particular forms of the loss functions under consideration. We therefore propose a generalized function approach.

The generalized function approach has the following advantages: (1) Generalized functions include ordinary functions that are locally integrable. (2) More importantly, in generalized function sense, generalized functions always have derivatives of all orders. This makes it possible to use Taylor expansion of generalized functions for these loss functions although they may not be necessarily smooth at some points (see \citealp{stankovic1996}).

\citet{phillips1991, phillips1995} are among the first in econometrics to develop a generalized function approach for studying asymptotic properties for LAD and M-estimators in the linear model setting. In particular, the two papers show Taylor expansion for generalized functions, establish (weak and strong) law of large numbers and central limit theorem in generalized function context. In addition, \citet{walsh2008} uses a generalized function approach in kernel estimation to estimate a density that may not be differentiable at some points. Notice also that the idea of generalized function (concretely, delta-convergent sequence) approaches already have been used in earlier studies in statistics, such as \citet{walter1979} and \citet{walter1981}, where the authors use the property of delta sequence converging to Dirac delta function, $\delta(x)$, to estimate density function. Moreover, the rationale behind the conventional kernel estimation method is also the delta-convergent sequence $K_h(x)\equiv h^{-1}K(x/h)\to \delta(x)$ as $h\to 0$ in distribution sense where $K(\cdot)$ is a density (the theorem in \citealp[p. 62]{kanwal1983}). Thus, if $g(\cdot)$ is continuous at $x_0$, $\int K_h(x-x_0)g(x)dx\to g(x_0)$ when $h\to 0$, while kernel estimation in econometrics and statistics is merely a sample version of this integral.

In summary, this paper makes the following contributions to the relevant literature. (i) Model \eqref{m1} accommodates three different types of time series vectors in a multiple-index form; (ii) M-estimation is developed for the unknown parameters that deals with a number of loss functions as special cases, such as OLS, LAD, quantile and expectile estimators; (iii) Non-differentiable loss functions are dealt with by a generalized function approach that provides a unified framework in M-estimation; (iv) The empirical findings show that the proposed model outperform its natural competitors in terms of understanding and evaluating stock return predictability from the points of view of both pseudo out-of-sample ${\rm R}^2$ and distributional quantiles; and (v) The finite-sample simulation results show that the proposed model and estimation theory works well numerically.

The rest of the paper is organized as follows. Section 2 introduces the estimation method and then the necessary assumptions before an asymptotic theory is established in Section 3. An empirical analysis about stock return predictability is given in Section 4. Simulation results are then given in Section 5. Section 6 concludes the main sections of this paper. Useful definitions and properties about generalized functions are given in Appendix A, and Appendices B and C include all the necessary mathematical technicalities.

\section{Assumptions for M--estimators}

Let $\rho(\cdot)$ be a convex nonnegative function defined on $\mathbb{R}$ that plays a role of loss function. Given the objective function $L_n(\theta,\beta)$ in equation \eqref{objfun}, we define the M-estimators $(\widehat\theta,\widehat\beta)$ of $(\theta_0,\beta_0)$ that satisfy
\begin{align}\label{estimator}
L_n(\widehat\theta,\widehat\beta)<\underset{\theta,\beta}{\arg\min} \ L_n(\theta,\beta)+o_P(\epsilon_n^2),
\end{align}
where $\epsilon_n\to 0$ as $n\to\infty$.
\medskip

Our theory relies on the assumption that $\rho(\cdot)$ is twice differentiable. When $\rho^{\prime} (\cdot)$ and/or $\rho^{\prime \prime}(\cdot)$ may not exist at some points, we shall invoke its derivatives in a generalized function sense, 
as discussed in Appendix A below. Indeed, our theory depends heavily on the theory of generalized functions.

\medskip

We now give some examples of the derivatives of generalized functions we discuss in the sequel.

\medskip

\noindent{\bf Examples}.\ \ The function given by
\begin{equation*}
  H(x)=\begin{cases}
  1, & x>0,\\
  0, & x<0,
  \end{cases}
\end{equation*}
is called Heaviside function. Its value at zero usually is defined as $1/2$, while sometimes it is taken as $c$ for $0<c<1$, written as $H_c(x)$. Simply by the definition of the derivative of generalized functions, $H'(x)= \delta(x)$ the so-called Dirac delta function.

\medskip

(a) Let $\rho(u)=|u|$ (LAD loss). Again, by the definition of the derivative of generalized functions, $\rho'(u)= -H(-u)+H(u)$, and thus $\rho''(u)=\delta(-x) +\delta(x)= 2\delta(x)$, where we use $\delta(-u)=\delta(u)$ given in \citet[p. 54]{kanwal1983}.

\medskip

(b) Let $\rho(u)=u(\tau-I(u<0))$ for some $0<\tau <1$ (Quantile loss). Similarly, we have $\rho'(u)=\tau-I(u<0)=\tau H(u)+(\tau-1) H(-u)$. Hence, $\rho''(u)=\tau \delta(u)-(\tau-1)\delta(-u)=\delta(u)$.

\medskip

(c) For Huber's loss function $\rho(u)$, it has the following derivative,
\begin{equation*}
\rho'(u)=\begin{cases}
  u, & |u|\le c,\\
  c\, \text{sgn}(u), & |u|>c.
  \end{cases}
\end{equation*}

Moreover, $\rho''(u)=I(|u|\le c)$ in generalized function sense.

\medskip

(d) For expectile loss function $\rho(u)=|\tau-I(u<0)|u^2$, it is a smooth function in ordinary sense with $\rho'(u)=2|\tau-I(u<0)|u$, whereas $\rho'(u)$ is not smooth at $u=0$. By definition, $\rho''(u)=2|\tau-I(u<0)|$.

\medskip

(e) For $L_p$ norm loss $\rho(u)=|u|^p$ with $p>1$, we have $\rho'(u)=p|u|^{p-1}\text{sgn}(u)$. However, if $1<p<2$, as an ordinary function $\rho''(u)$ does not exist at $u=0$. In generalized function sense, $\rho''(u)=p(p-1)|u|^{p-2}$, while the meaning of generalized function $|u|^{p-2}$ is discussed in \citet{gelfand1964} and \citet{kanwal1983} for example.

\medskip

Before stating assumptions needed for our theoretical development, we introduce some notations. In what follows, $\|\cdot\|$ signifies Euclidean norm for vector or element-wise norm for matrix; $\int$ means an integral (maybe multiple integral) over entire real set $\mathbb{R}$ ( $\mathbb{R}^p$).

\medskip

\noindent {\bf Assumption A}\ \
\begin{enumerate}[({A.}1)]
\item \emph{Let $\{\eta_j, j\in \mathbb{Z}\}$ be a vector sequence of $d_1$-dimensional independent and identically distributed (iid) continuous random variables satisfying $E\eta_0=0$, $E\eta_0\eta_0'=I_{d_1}$ and $E\|\eta_0\|^4<\infty$. Let $\varphi(u)$ be the characteristic function of $\eta_0$ satisfying $\int \|u\|\, |\varphi(u)| du <\infty$.}
   \item \emph{Let $\{w_t\}$ be a linear process defined by $w_t=\sum_{j=0}^\infty A_j \eta_{t-j}$, where $\{A_j\}$ is a sequence of square matrices, $A_0=I_d$, and $\sum_{j=0}^\infty j \|A_j\| <\infty$. Moreover, suppose that $A:=\sum_{j=0}^\infty A_j$ has full rank.}
  \item \emph{Let $x_{t}=x_{t-1}+w_{t}$ for $t\geq 1$ with $x_{0}=O_P(1)$.}
\end{enumerate}

The structure of the unit root process $x_t$ is commonly imposed in the literature since $w_t$ has a considerably general form that covers many special cases. See \citet{phillips1999,phillips2001} and \citet{dgd2016}. It is straightforward to calculate that $d_t^2:=\mathbb{E} (x_t x_t^\top)= AA^\top t(1+o(1))$, and it is well known that, for $r\in [0,1]$, $n^{-1/2}x_{[nr]}\to_DB(r)$ as $n\to\infty$ where $B(r)$ is a Brownian motion on $[0,1]$ with covariance $AA^\top$.

\medskip

\noindent \textbf{Assumption B}.

\begin{enumerate}[({B.}1)]
\item \emph{Let $h(r, v)$ be a $d_2$-dimensional vector of continuous functions in $r$ on $[0,1]$.}

\item \emph{Suppose that $v_t=q(\eta_t,\cdots, \eta_{t-d_0+1}; \tilde{v}_t)$ for $d_0\geq 1$, where $\{\tilde{v}_t, 1\le t\le n\}$ is also a vector sequence of strictly stationary and $\alpha$-mixing time series, and independent of $\eta_j$'s introduced in Assumption A, and the vector function $q(\cdots)$ is measurable such that
    \begin{equation*}
    \sup_{r\in [0,1]}\e\|\dot{g}_2^2(\beta_0^\top h(r, v_1))\; h(r, v_1)h(r, v_1)^\top\|^4<\infty.
    \end{equation*}
    Moreover, the $\alpha$-mixing coefficient $\alpha(k)$ of $\{\tilde{v}_t, 1\le t\le n\}$ satisfies $\sum_{k=1}^\infty \alpha^{1/2} (k) <\infty$.}
\end{enumerate}

Basically, $\{v_t\}$ is a vector sequence of strictly stationary variables and $v_t$ can be correlated with $x_t$ through a finite number of innovations of $\eta$'s. The existence of the fourth moment and the condition on the mixing coefficients guarantee the convergence of the average of $\dot{g}_2^2(\beta_0^\top z_t)z_t z_t^\top$ in probability by virtue of Davydov's inequality, as shown in Appendix B. All these are quit common but in this paper relies on a known function $h(\cdot)$, so we impose in B.1 the continuity on $h(\cdot, v)$ to make sure that $h(r,v)$ is bounded in $r\in [0,1]$ for each $v$. Let $\xi_{1t} = \rho'(e_t)$, $\xi_{2t} = \rho'(e_t)z_t$ and $\xi_{3t} = \rho''(e_t)$.

\medskip

\noindent \textbf{Assumption C}. \
\begin{enumerate}[({C.}1)]
\item \emph{Suppose that $\rho(\cdot)$ is positive convex, and has its only minimum at zero; moreover, it is locally integrable function on the real line.}
\item \emph{Suppose that $(\xi_{1t},\mathcal{F}_{t})$ is a martingale difference sequence where $\mathcal{F}_t$ is a $\sigma$-field such that $x_t$ and $v_t$ are adapted to $\mathcal{F}_{t-1}$. Moreover, $\mathbb{E}[\xi_{1t}^2|\mathcal{F}_{t-1}]=a_1$ and $\mathbb{E}[\xi_{3t}|\mathcal{F}_{t-1}]=a_2$ almost surely for some $0<a_1,a_2<\infty$. In addition, $\max_{1\le t\le n} \mathbb{E}[\xi_{1t}^4| \mathcal{F}_{t-1}]<\infty$ uniformly in $n\geq 1$.}

\item \emph{Suppose that, as $n\to\infty$, for $r\in [0,1]$,}
\begin{equation*}
  \left(\frac{1}{\sqrt{n}}\sum_{t=1}^{[nr]} \xi_{1t}, \frac{1}{\sqrt{n}}\sum_{t=1}^{[nr]} \xi_{2t}, \frac{1}{\sqrt{n}}\sum_{t=1}^{[nr]}\eta_t\right)\Rightarrow (U_\rho(r), V_\rho(r), W(r)),
\end{equation*}
\end{enumerate}
\emph{which is a Brownian motion of $(d+1)$ dimension, where $U_\rho(r), V_\rho(r)$ and $W(r)$ have variance and covariance, $a_1$, $a_1 \int \e(h(r, v_1)h(r, v_1)^\top)dr$ and $I_{d_1}$, respectively}.

\medskip

Condition (C.1) is mild that ensures the eligibility of $\rho(\cdot)$ as a generalized function, apart from being a loss function. In Assumption (C.2), $\rho'(\cdot)$ and $\rho''(\cdot)$ are the first two derivatives of $\rho(\cdot)$ possibly in the generalized function sense, while they are the same as the derivatives in ordinary sense whenever exist.  The filtration in (C.2) can be taken as $\mathcal{F}_{t}=\sigma(e_1,\cdots, e_t; x_s, z_s, s\le t+1)$ that is general than series independence between $e_t$ and $(x_t,z_t)$.  When $\rho'(\cdot)$ exists in ordinary sense in (C.2), the martingale difference sequence condition renders $\mathbb{E}[\xi_{1t}|\mathcal{F}_{t-1}]=0$ that is easily fulfilled for the loss functions of OLS, LAD, $p$-norm ($p>1$) and Huber's when the conditional distribution of $e_t$ given $\mathcal{F}_{t-1}$ is symmetric, because these functions have odd derivatives. Also, this is the same as $F_e(0|\mathcal{F}_{t-1}) =\tau$ for quantile loss (see, for example, Assumption 2 in \citealp{degui2019}). Some authors obtain this condition by adjusting the intercept that we also have to deal with when it is violated, like \citet[p. 124]{hexuming2000} and \citet[p. 251]{xiao2009a}. In addition, the positiveness $\mathbb{E}[\xi_{3t}|\mathcal{F}_{t-1}]=a_2>0$ is due to the convexity of the loss function under consideration. In particular, $a_2=f_e(0)>0$ for both check loss and LAD loss, while for Huber's loss, $a_2= P(|e_t|\le c)>0$ in exogenous situation.

When $\rho'(\cdot)$ and $\rho''(\cdot)$ in (C.2) are genuinely generalized functions, the conditional expectations can be defined as
\begin{equation*}
\e[\xi_{1t}|\mathcal{F}_{t-1}]=-\int \rho(u)f'_{t,e}(u)du \ \ \mbox{and} \ \ \e[\xi_{3t}|\mathcal{F}_{t-1}]=\int \rho(u)f''_{t,e}(u)du
\end{equation*}
under some conditions on the conditional density $f_{t,e}(u)$ of $e_t$ given $\mathcal{F}_{t-1}$. Though $\rho(\cdot)$ may not be differentiable at some isolated points, $\mathbb{E}[\xi_{1t}|\mathcal{F}_{t-1}]$ and $\mathbb{E}[\xi_{3t}|\mathcal{F}_{t-1}]$ are well-defined given that the conditional density is smooth and is a member in $S$ (the rapid decay test function space, see Appendix A); this is due to the use of generalized functions that extend the ordinary derivatives. The condition D.2 in \citet[p. 125]{hexuming2000} defines $\mathbb{E}[\xi_{3t}]$ in the same fashion as for non-differentiable loss functions.

Note that (C.3) is quite commonly encountered in the cointegrating regression literature when there is no additional stationary regressor involved (see \citealp{phillips2001}, \citealp{xiao2009a} and \citealp{qiying2018}). Nevertheless, here the subscript $\rho$ indicates that the joint convergence relies on the loss function. See, for example, equation (3) of \citet{phillips1995} shows the functional invariance principle for the derivative of LAD loss, i.e. sign$(\cdot)$.

For simplicity of notation we would mostly suppress these subscript, viz., denote $U_\rho(r)$ by $U(r)$, and $V_\rho(r)$ by $V(r)$, if there is no confusion raised.

Note that (C.3), in particular, ensures that as $n\rightarrow \infty$
\begin{equation}\label{bmotion}
 \left(\frac{1}{\sqrt{n}}\sum_{t=1}^{[nr]} \xi_{1t}, \frac{1}{\sqrt{n}}\sum_{t=1}^{[nr]} \xi_{2t}, \frac{1}{\sqrt{n}}x_{[nr]}\right)\Rightarrow (U(r),V(r), B(r)),
\end{equation}
where $B(r)$ has a different covariance function from that of $W(r)$.

\medskip

Because of the involvement of the unit root process $x_t$, we need a similar classification on the function classes, such as $H$-regular and $I$-regular in \citet{phillips1999,phillips2001}. However, in our paper we only consider a simpler version of $H$-regular class though a general form can be adopted straightforwardly.

\begin{defn}
\emph{A function $f(\cdot)$ defined on $\mathbb{R}$ is called $H$-regular with homogeneity order $\nu(\cdot)$ if $f(\lambda x)=\nu(\lambda)f(x)$ for all $\lambda>0$ and any $x$. A function $f(\cdot)$ defined on $\mathbb{R}$ is called $I$-regular if $\int |f|<\infty$.}
\end{defn}

The class of $H$-regular functions mainly contains power functions while the class of $I$-regular functions has all integrable functions on the entire real line as its member, such as probability density functions. More detailed discussion on these definitions can be found in \citet{phillips1999, phillips2001}.

\medskip

\noindent{\bf Assumption D}\; \; \emph{Let both $g_1(u)$ and $g_2(v)$ be continuously differentiable up to the second order. Suppose further that}
\begin{enumerate}[({D.}1)]
  \item \emph{$g_1(u)$, $\dot{g}_1(u)$ and $\ddot{g}_1(u)$ are $H$-regular with homogeneity order $\nu(\cdot)$, $\dot{\nu}(\cdot)$ and $\ddot{\nu}(\cdot)$, respectively, such that $\lim_{\lambda\to +\infty} \dot{\nu}(\lambda)/\nu(\lambda)=0$, and $\lim_{\lambda\to +\infty} \ddot{\nu}(\lambda)/\dot{\nu} (\lambda) =0$.}


  \item \emph{$g_1(u)$ and $g_1^2(u)$ are $I$-regular. Moreover, $\dot{g}_1(u)$, $\dot{g}_1^2(u)$, $u\dot{g}_1(u)$, $[u\dot{g}_1(u)]^2$, $\ddot{g}_1(u)$ and $u\ddot{g}_1(u)$ are all $I$-regular.}
\end{enumerate}

\medskip

It is possible to relax the restriction on the function forms but they are the benchmark. For example, mixtures of $H$-regular and $I$-regular in terms of argument $u$ are feasible for the development in what follows, while we do not pursue it since this would make notation much complicated. 

\section{Asymptotic theory}

In the process of establishing our asymptotic theory, we shall consider in the proofs that each of the generalized derivatives of $\rho(\cdot)$ is a limit of regular sequence under the setting of generalized functions (see, for example, \citet{phillips1991, phillips1995}).

\begin{theo}\label{th1}
Under Assumptions A-C and D(1), as $n\to\infty$,
\begin{align*}
  \begin{pmatrix}
  \dot{\nu}(\sqrt{n})n(\widehat \theta-\theta_0)\\
  \sqrt{n}(\widehat \beta-\beta_0)
  \end{pmatrix}
  \to_D\begin{pmatrix}
  [A_1-A_3^\top A_2A_3]^{-1}(b_1-A^\top_3 A_2^{-1}b_2)\\
  A_2^{-1}(b_2-A_3[A_1-A_3^\top A_2A_3]^{-1}(b_1-A_3^\top A_2^{-1}b_2))
  \end{pmatrix},
\end{align*}
where
\begin{align*}
  b_1= & \int_0^1\dot{g}_1(B(r)^\top\theta_0)B(r) dU(r), \\
  b_2\sim&N(0, a_1\Sigma), \ \text{with}\ \Sigma=\int_0^1\mathbb{E}[\dot{g}_2^2(\beta_0^\top h(r, v_1))\, h(r,v_1) h(r,v_1)^\top] dr,\\
  A_1= & a_{2}\int_0^1\dot{g}_1(B(r)^\top\theta_0)B(r)B(r)^\top dr, \ \ \ A_2= a_2\Sigma,\\
  A_3=&a_{2}\int_0^1 \dot{g}_1(B(r)^\top\theta_0) [\mathbb{E}\dot{g}_2(h(r,v_1)^\top\beta_0) h(r,v_1)]B(r)^\top dr,
\end{align*}
and positive numbers $a_1$ and $a_2$ are given in Assumption C.
\end{theo}
\medskip

{\bf Remark}.\ \ It can be seen from the theorem that the imposition of $H$-regular functions $g_1(u)$ and its derivatives delivers a fast rate of convergence for $\widehat\theta$ due to the divergence of unit root process $x_t$. This is comparable with the result in \citet{phillips2001}. On the other hand, the estimator $\widehat\beta$ possesses a usual square-root-$n$ rate, although the regressor $z_t$ has a deterministic trending component. Notice that the positive definiteness of $\Sigma$ is easily satisfied as long as $v_1$ is a continuous variable. More importantly, it is readily seen that the different loss functions take a role in the asymptotic limits through $a_1$ and $a_2$ that may affect the covariance but not the rates of convergence.

In order to make a comparison with the literature, here we simply suppose the model is exogenous. (1) When $\rho(u)=|u|$, $\rho'(u)=-H(-u)+H(u)$ with $H(\cdot)$ being Heaviside function and $\rho''(u)=2\delta(u)$. We have $a_1=\mathbb{E}[\rho'(e_t)]^2=1$ and $a_2=\mathbb{E}[\rho''(e_t)] =2f_e(0)>0$ where $f_e(\cdot)$ is the density of $e_t$. Moreover, if $g_1(u)=u$ and $g_2(v)\equiv 0$, the model reduces to a linear cointegrating model. Our result implies that
$$n(\widehat \theta-\theta_0)\to_D [2f_e(0)]^{-1}\left[\int_0^1B(r)B(r)^\top dr\right]^{-1} \int_0^1 B(r)dU(r),$$
that is the result of \citet{phillips1995} without correlation between $x_t$ and $e_t$; if $g_2(v)=v$ and $g_1(u)\equiv 0$, $h(r, v)\equiv v$ the model reduces to a linear model with stationary regressors. Our result also naturally covers that $\sqrt{n}(\widehat \beta-\beta_0)\to_D[2f_e(0)]^{-1} [\mathbb{E}(v_1v_1^\top)]^{-1/2} N(0,I_{d_2})$ as shown in \citet{pollard1991}.

(2) When $\rho(u)$ is the quantile loss, $\rho'(u)=\tau H(u)+(\tau-1)H(-u)$ and $\rho''(u)=\delta(u)$. We have $a_1=\mathbb{E}[\rho'(e_t)]^2 =\tau(1-\tau)$ and $a_2=\mathbb{E}[\rho''(e_t)]=f_e(0)>0$ where $f_e(\cdot)$ is the density of $e_t$. Further, if $g_1(u)=u$ and $g_2(v)\equiv 0$, the model reduces to a linear cointegrating model. Our result implies that $n(\widehat \theta-\theta_0)\to_D [f_e(0)]^{-1} [\int_0^1B(r)B(r)^\top dr]^{-1} \int_0^1 B(r)dU(r)$ that is the result of \citet{xiao2009a} without correlation between $x_t$ and $e_t$; if $g_2(v)=v$ and $g_1(u)\equiv 0$, $h(r,v)\equiv v$ the model reduces to a linear model with stationary regressor. Our result implies that $\sqrt{n}(\widehat \beta-\beta_0)\to_D\sqrt{\tau(1-\tau)}[f_e(0)]^{-1} [\mathbb{E}(v_1v_1^\top)]^{-1/2} N(0,I_{d_2})$, which is the result of quantile regression for linear model. See \citet{koenker1978}. $\Box$

\medskip

While $\Sigma$ can easily be consistently estimated by the function $h$ and observations of $v_t$, we need to construct consistent estimators for both $a_1$ and $a_2$ in order to use the asymptotic limits for statistical inferences. These estimators are given by $\widehat{a}_1=\frac{1}{n}\sum_{t=1}^n [\rho'(\widehat{e}_t)]^2$ and $\widehat{a}_2=\frac{1}{n}\sum_{t=1}^n \rho''(\widehat{e}_t$), respectively, where $\widehat{e}_t=y_t-g(\widehat{\theta}^{\,\top} x_t, \widehat{\beta}^{\,\top} z_t)$ for $1\leq t\leq n$. We then have the following corollary.

\begin{cor}\label{cor1}
Let the conditions of Theorem \ref{th1} hold. If, in addition, $\{[\rho'(e_t)]^2\}$ and $\{\rho''(e_t)\}$ are mean ergodic, then we have $\widehat{a}_1\to_Pa_1$ and $\widehat{a}_2\to_Pa_2$ as $n\to\infty$.
\end{cor}

The ergodicity for the two sequences is quite weak and can therefore be fulfilled under several sufficient conditions, such as requiring $e_t$ be a martingale difference sequence or a strictly stationary and $\alpha$-mixing sequence. When the derivations of the loss function $\rho(u)$ do not exist in ordinary sense, its regular sequence may be used for $\rho(u)$ (see Appendix A and \citet[p 918]{phillips1995}).
\medskip

Now we consider the regression function $g(u,v)=g_1(u)+g_2(v)$ where both $g_1(u)$ and $g_2(v)$ are smooth but $g_1(u)$ is integrable on the entire real line stipulated by Assumption D(3). This assumption will change the asymptotic theory drastically. The limit theory depends on the local time of some scalar Brownian motion $W(r)$ defined as
\begin{equation*}
  L_W(p,s)=\lim_{\epsilon\to 0}\frac{1}{2\epsilon} \int_0^pI(|W(r)-s|<\epsilon) dr,
\end{equation*}
which measures the sojourning time of $W(r)$ around $s$ during the time period $(0,p)$. In our context, the scalar Brownian motion may be $B_1(r):=\theta_0^\top B(r)$ where $B(r)$ is defined by \eqref{bmotion}.

Another feature of the limit theory about this regression function is that we derive in a new coordinate system first, then recover the limit in the original coordinate. To do so, using $\theta_0$ that is a nonzero but not necessarily a unit vector, we construct an orthogonal matrix $P=(\theta_0/\|\theta_0\|, P_1)_{d_1\times d_1}$ to rotate all the vectors of interest, including $\theta_0$, $x_t$ and $B(r)$. The relevant rotated vectors are $x_{1t}\equiv\theta_0^\top x_t$ and $x_{2t}\equiv P_1^\top x_t$ that after normalization converge jointly to $B_1(r)\equiv\theta_0^\top B(r)$ and $B_{2}(r)\equiv P_1^\top B(r)$; and we shall denote by $L_1(p,s)$ the local time of $B_1(r)$.

\begin{theo}\label{th2}
Under Assumptions A-C and D(2), as $n\to\infty$ we have
\begin{align*}
\begin{pmatrix}\widehat\alpha\\ \widehat \gamma \end{pmatrix}=\begin{pmatrix} \sqrt[4]{n}\frac{1}{\|\theta_0\|^2}\theta_0^\top (\widehat\theta-\theta_0)\\
\sqrt[4]{n}^3P_1^\top (\widehat\theta-\theta_0)\\
\sqrt{n}(\widehat\beta-\beta_0)
\end{pmatrix}\to_D\frac{\sqrt{a_1}}{a_2}M^{-1/2}N(0,I_d),
\end{align*}
where $P\equiv(\theta_0/\|\theta_0\|, P_1)$ is an orthogonal matrix, $a_1$ and $a_2$ are positive constants given in Assumption C, $M=\text{diag}(M_1, \Sigma)$ is a symmetric positive definite $d\times d$ matrix in which $M_1=(m_{ij})$ is given by
\begin{align*}
m_{11}=&\int [u\dot{g}_1(u)]^2du\; L_1(1,0), & m_{12}=&\int_0^1 B_2^\top(r)dL_1(r, 0) \int u \dot{g}_1^2(u)du,\\
m_{21}=&m_{12}, & m_{22}=&\int_0^1B_2(r)B_2(r)^\top dL_1(r,0)\int \dot{g}_1^2(u)du,
\end{align*}
and $\Sigma$ is the $d_2\times d_2$ matrix defined defined in Theorem 2.1.
\end{theo}

Noting that $M$ is a block diagonal matrix, the result in Theorem \ref{th2} implies
\begin{align}
D_nJ& P^\top (\widehat\theta-\theta_0)=\begin{pmatrix} \sqrt[4]{n}\frac{1}{\|\theta_0\|^2}\theta_0^\top \\
\sqrt[4]{n}^3P_1^\top
\end{pmatrix}(\widehat\theta-\theta_0)\to_D\frac{\sqrt{a_1}}{a_2}
N(0,M^{-1}_1),\label{int1}\\ \intertext{and}
&\sqrt{n}(\widehat\beta-\beta_0)\to_D\frac{\sqrt{a_1}}{a_2}N(0,\Sigma^{-1}), \label{int2}
\end{align}
as $n\to\infty$ where $J=\text{diag}(1/\|\theta_0\|, I_{d_1-1})$.

Notice also that in the condition of D.2, $g_2(v)$ is the same as in D.1. However, the convergence of \eqref{int2} has not been bothered by the first additive component $g_1(v)$, because the convergence is similar to the existing literature except our setting for the regressor is more complicated (see \citealp{pollard1991} and \citealp{koenker1978}). By sharp contrast, the asymptotic limit of $\widehat\beta$ in Theorem 2.1 is affected by the limit of the regressor in the first component function. This is due to the drastic different behavior between the H-regular and I-regular functions of unit root processes. See, for example, \citet{phillips2001} and \citet{dgd2016}.

The convergence of $\eqref{int1}$ is similar to Theorem 3.1 in \citet{dgd2016} where $J=I_d$ due to the identification condition $\|\theta_0\|=1$ in semiparametric single-index model, $a_1=4\sigma_e^2$ and $a_2=2$ because $\rho(u)=u^2$ in the paper. It is also similar to Theorem 5.1 in \citet{phillips2001} where the regressor is univariate but the parameter is multivariate. In addition, when the loss function reduces to LS loss, the convergence of $\eqref{int1}$ is the same as Theorem 2 in \citet{phillips2000}.

If one is concerned about the estimate of the matrix $P_1$, note that $P_1$ is an orthogonal system of $\theta_0^{\bot}$ (orthogonal complement space), so that once $\widehat\theta$ is available, we may define $\widehat P_1= \widehat\theta^{\bot}$. Normally, one does not need to rotate the coordinate system in practice that is introduced to facilitate our theory only.

Using the block representation of $M_1=(m_{ij})_{2\times 2}$, we have $M_1^{-1}=(\tau_{ij})_{2\times 2}$ with
\begin{equation*}
\tau_{11}=(m_{11}-m_{12}m_{22}^{-1}m_{21})^{-1}\ \ \ \text{and}\ \ \ \tau_{22}=(m_{22}-m_{21}m_{11}^{-1}m_{12})^{-1}.
\end{equation*}

The following corollary recovers the asymptotic distribution of $\widehat\theta$ from the limit of its rotation.

\begin{cor}\label{cor2}
Under the conditions in Theorem \ref{th2}, we have as $n\to\infty$
\begin{equation*}
\sqrt[4]{n}(\widehat\theta-\theta_0)\to_DN(0, \tau_{11}\theta_0\theta_0^\top).
\end{equation*}
\end{cor}

Though indicated by \eqref{int1} that the coordinates of $\widehat\theta$ on $\theta_0$ and $P_1$ have rates $\sqrt[4]{n}$ and $\sqrt[4]{n}^3$, respectively, the estimator $\widehat\theta$ eventually has slower rate $\sqrt[4]{n}$. This is the same as \citet{dgd2016} and one may find more detailed explanation therein. Interestingly notice that since the model in \citet{dgd2016} is nonparametric, $\theta_0$ is identified up to its direction, so the condition $\|\theta_0\|=1$ is imposed. Thus, the authors normalize the estimator $\widehat\theta$ to be a unit vector and find that the normalized estimator has a much faster rate than $\sqrt[4]{n}$. By contrast, no identification condition is needed in this paper as the model is parametrically nonlinear. Consequently we would not be able to achieve any rates faster than what we have obtained. However, if one could know the length of $\theta_0$, say $\|\theta_0\|=q$, the estimator $q\;\widehat\theta/\|\widehat\theta\|$ would converge to $\theta_0$ with a rate of an order $\sqrt[4]{n}^3$, basically because the normalization stretches $\widehat\theta$ directly on the unit circle with radius $q$ (see, Theorem 3.2 in \citet{dgd2016}, for example). Moreover, one may have the estimates of $a_1$ and $a_2$ similarly to Corollary \ref{cor1} and that of the local time in the limit similarly to \citet{dgd2016}. The details are omitted due to the similarity.

\medskip

Before we conclude this section, we briefly discuss another important issue related to model \eqref{m1}. If the model has conditional heteroscedasticity, such as $e_t=\sigma(\vartheta_0^\top x_t, \pi_0^\top z_t) \varepsilon_t$, where $\varepsilon_t$ is independent of $(z_t,x_t)$, and if $\sigma(u,v)=\sigma_1(u)\sigma_2(v)$ is known and separable, we will have
\begin{align*}
\log(e_t^2)=&\log(\sigma_1^2(\vartheta_0^\top x_t))+
\log(\sigma_2^2(\pi_0^\top z_t))+\log(\varepsilon^2_t)\\
=&\mu_0+\log(\sigma_1^2(\vartheta_0^\top x_t))+
\log(\sigma_2^2(\pi_0^\top z_t))+\zeta_t,
\end{align*}
where $\zeta_t$ is the centralized version of $\log(\varepsilon^2_t)$ with $\mu_0=\mathbb{E}[\log(\varepsilon^2_t)]$.

In this case, the proposed estimation method is readily applicable for us to derive the corresponding M-estimators of $\vartheta_0$ and $\pi_0$ after we replace $e_t$ by $\widehat{e}_t=y_t-g_1(\widehat{\theta}^{\;\top} x_t)-g_2(\widehat{\beta}^{\;\top} z_t)$. Under some same assumptions on $\log(\sigma_1^2(u))$ and $\log(\sigma_2^2(u))$ as in Assumption D, we may establish asymptotic properties corresponding to the above theorems and corollaries for the estimators of $\vartheta_0$ and $\pi_0$.

\section{Stork return predictability}

To demonstrate the practical relevance and superiority of our proposed model and estimation method over some natural competitors, we investigate its applicability in stock return predictability. It is common to use a linear predictive mean regression in the literature, which has led to considerable disagreements in the empirical finding as to whether stock returns are predictable or not (\citealp{CT2008}; \citealp{WG2008}). Recently, \citet{KAP2015a} use a nonparametric mean regression, \citet{Lee2016} and \citet{FL2019} consider a quantile linear regression, and \citet{TLW2021} adopt a nonparametric quantile framework to re-investigate this important issue. We shall demonstrate the nonlinearity in return predictability through the use of the proposed multiple-index model with both stationary and nonstationary predictors.

The data sets to be used for return prediction are obtained from \citet{WG2008}. They have been extensively used in the predictive literature, including the recent balanced predictive mean regression model by Ren et al. (2019), the linear quantile predictive regression by \citet{Lee2016} and \citet{FL2019}, the linear prediction with cointegrated variables by \citet{kasy2020}, the LASSO predictive regression of
\citet{LSG2021}, and the nonparametric quantile predictive regression of \citet{TLW2021}, among others. The monthly data spans from January 1927 to December 2005. The dependent variable, excess stock return, is defined as the difference between the S\&P 500 index return, including dividends and the one month Treasury bill rate. The nonstationary (persistent) predictors include dividend-price ($dp$), dividend-payout ratio ($de$), long term yield ($lty$), book to market ($bm$) ratios, T-bill rate ($tbl$), default yield spread ($dfy$), net equity expansion ($ntis$), earnings-price ($ep$), term spread ($tms$), while the stationary variables include default return spread ($dfr$), long term rate of return ($ltr$), stock variance ($svar$) and inflation ($infl$). The first-order serial correlations of these variables and their time series plot are available from, for example, \citet{LSG2021}. For detailed description on each series and the data construction, please refer to \citet{WG2008}.

To measure the performance of the predictive regression models under discussion, we use the pseudo out-of-sample $R^2$, defined as
$$PR^2=\frac{\sum_{t=T_1+1}^{T_2} \rho(\widehat e_t)}{\sum_{t=T_1+1}^{T_2} \rho(\overline{e}_t)},$$
where $\widehat{e}_t$ and $\overline{e}_t$ are the respective out-of-sample prediction errors at time $t$ for a given model and that for the base model with only a constant predictor, for the forecasting sample spanning from $T_1$ to $T_2$. Four loss functions are entertained for $\rho(\cdot)$: (L1) squared errors loss $\rho(e)=e^2$; (L2) absolute errors loss $\rho(e)=|e|$,  (L3) Huber's loss $\rho_\delta (e)=\frac{1}{2}e^2\cdot 1\{|e|\leq \delta\}+\delta\cdot (e-\frac{1}{2} \delta) \cdot 1\{|e|> \delta\}$ with $\delta=1.25$; and (L4) quantile loss $\rho_\tau (e)=e\cdot (\tau-I(e<0))$. Positive $PR^2$ indicates the better performance of the given model over the base model, and the larger the value is, the better the performance is.

For space limitation, we follow \citet{KAP2015a}, \citet{Lee2016}, and \citet{TLW2021} to consider two subsamples: (i) the period spanning from January 1927 to December 2005, where significant predictability has been discovered, and (ii) the tranquil period starting from January 1952 to December 2005, where mixed evidences are found for predictability. We consider a rolling-window scheme for the out-of-sample return prediction and let the rolling window in-sample size (RWS) be 120 (10 years), 240 (20 years), and 360 (30 years), respectively. For example, the forecast starts from January 1927 and ends at December 2005 for the first subsample when RWS equals 120.

In contrast to \citet{TLW2021}, our proposed model allows the use of multivariate nonstationary predictors together with those stationary ones to capture nonlinearity in the return predictability. There are numerous choices of combinations of the nonstationary and stationary predictors. For demonstration purposes, we use the stationary predictor $z=\{inf, svar\}$, and choose the nonstationary predictors as either $x=\{bm,lty\}$ or $x=\{dp,tms\}$.  When $x=\{bm,lty\}$, we consider an exponential link function $g_1(u)=u\cdot e^{-0.4u^2}$ and a linear link $g_2(v)=v$ outside of the nonstationary (resp. stationary) and stationary (resp. nonstationary) index variables, respectively. This amounts to four possible combinations. When $x=\{dp,tms\}$, we alternatively consider the normal cumulative distribution function $\Phi$ to capture possible nonlinearity. The logistic distribution function has also been experimented in this case and makes little difference to subsequent main findings, the results of which are therefore not reported below.
The extensive analysis with a comprehensive set of predictors noted above and all possible specifications of the nonlinear link functions remains an interesting but a challenging work, which is left as future study.

Tables \ref{PR2_1} and \ref{PR2_2} collect the pseudo out-of-sample $R^2$ for excess return prediction with the above two sets of selected predictors and specifications of the link functions, respectively, for the loss functions specified in (L1)-(L3). For each loss function and each RWS, the combination of the link function that produces the largest $PR^2$ is bolded. There are several interesting findings. First, noticeable nonlinearity exists in return prediction, for both choices of the nonstationary predictors, and for both subsamples considered. It is especially found that for the squared error loss and the Huber's loss, the linear specification is always worse than the nonlinear alternatives. When $x=\{bm,lty\}$, nonlinearity has been discovered for both the stationary index and the nonstationary component. However, when $x=\{dp,tms\}$, only the nonstationary index presents nonlinearity in return prediction. Second, nonlinear predictability is found in forecasting return when linear predictability is not. The linear predictive regression (nearly) outperforms the constant model without any predictor when $x=\{bm,lty\}$, but the former underperforms the latter when $x=\{dp,tms\}$. Although the linear model does not reveal predictability in this case, it is uniformly observed that the nonlinear link function specification improves over the constant model for both choices of subsamples. Third, the conventional ``tranquil'' period 1952-2005 reveals strong nonlinear predictability, which is quite comparable to the 1927-2005 subsample. This finding is in contrast to the declining linear predictability discovered by \citet{CY2006}, or the weak nonlinear predictability with only a single nonstationary predictor by \citet{KAP2015a}. To summarize, significant return predictability has been revealed in the nonlinear predictive regression with both stationary and nonstationary indices we introduced.

As the recent literature has witnessed a growing interest in investigating the return prediction in quantiles (Lee, 2016; Fan and Lee, 2019; Tu, et al., 2021), we next study the performance of our proposed model in predicting the return quantiles. We consider the quantile level $\tau=0.05,0.1,0.2,\ldots,0.9,0.95$. The analysis will only focus on the nonlinear effect of the nonstationary index as discovered from Tables \ref{PR2_1} and \ref{PR2_2}. Therefore, we only present the results for $x=\{dp,tms\}$ and $z=\{inf, svar\}$ in Table \ref{PR2_3} at all the quantile levels specified above, for the two subsamples mentioned earlier. The nonlinear specification for the nonstationary component includes not only the normal distribution function, but also the popularly used logistic distribution function (Logit), as a robustness check. From Table \ref{PR2_3}, it is observed that linear quantile predictability seems to only exist for the lower quantile levels, such as those at $\tau=0.05,0.1$ and $0.2$. However, nonlinear quantile predictability prevails at all quantile levels. Even at the lower quantiles, the proposed nonlinear index model outperforms the linear model except for a very few cases. To compare the normal with logistic distributions, the latter seems to enjoy better performance for both sample periods. In addition, it seems that the nonlinear predictability becomes stronger as the quantile level moves to the two ends, especially for the first subsample.

\begin{table}[htbp]
\centering
\caption{Pseudo Out-of-sample $R^2$ for excess return prediction with $x=\{bm,lty\}$ and $z=\{inf, svar\}$.}\label{PR2_1}
\begin{tabular}[t]{lc l rrr l rrr}
\toprule
&&\multicolumn{8}{c}{Panel A: January 1927-December 2005}\\
\cline{2-10}
                       &        &&\multicolumn{3}{c}{$g_2(v)=v$}       &&  \multicolumn{3}{c}{$v\cdot e^{-0.4v^2}$}      \\
                       \cline{4-6}\cline{8-10}
$g_1(u)=$                      &RWS&& L1         &  L2        &  L3        &&  L1        &  L2         &   L3      \\
                      \hline
$u$                & 120 && 0.0306  &  \bf{0.0319}  &  0.0258 &&  0.0302  &  0.0276  &   0.0312\\
                      & 240 && 0.0658  &  0.0326  &  0.0763 &&  0.0601  &  0.0321  &   0.0747\\
                      & 360 && 0.0657  &  0.0076  &  0.0784 &&  \bf{0.0672}  &  0.0088  &   \bf{0.0797}\\
$u\cdot e^{-0.4u^2}$&
                       120 && 0.0346  &  0.0243  &  0.0341 &&  \bf{0.0364}  &  0.0266  &  \bf{0.0367}\\
                       &240 && 0.0788  &  \bf{0.0328}  &  \bf{0.0773} &&  \bf{0.0848}  &  0.0322  &  0.0723\\
                       &360 && 0.0325  &  \bf{0.0217}  &  0.0791 &&  0.0311  &  \bf{0.0217}  &  0.0780\\
\midrule
&&\multicolumn{8}{c}{Panel B: January 1952-December 2005}\\
\cline{2-10}
$u$    & 120 && -0.0274 &   \bf{0.0014} &  -0.0213 &&  -0.0255 &  -0.0063 &  -0.0144\\
                      & 240 && 0.0162  &  0.0063 &   0.0394 &&  0.0112 &   0.0063 &   0.0370\\
                      & 360 && 0.0177 &  -0.0331 &   0.0343 &&  \bf{0.0193} &  -0.0315 &   0.0363\\
$u\cdot e^{-0.4u^2}$&
                       120 && -0.0136 &  -0.0063 &  -0.0129 &&  \bf{-0.0134} &  -0.0031 &  \bf{-0.0077}\\
                       &240 && 0.0370  &  \bf{0.0093} &   \bf{0.0395} &&  \bf{0.0442} &   0.0086  &  0.0318\\
                       &360 && -0.0337 &  \bf{-0.0138}  &  \bf{0.0427} &&  -0.0355 &  -0.0139  &  0.0417\\
\bottomrule
\end{tabular}
\end{table}

\begin{table}[htbp]
\centering
\caption{Pseudo Out-of-sample $R^2$ for excess return prediction with $x=\{dp,tms\}$ and $z=\{inf, svar\}$.}\label{PR2_2}
\begin{tabular}[t]{lc l rrr l rrr}
\toprule
&&\multicolumn{8}{c}{Panel A: January 1927-December 2005}\\
\cline{2-10}
                       &        &&\multicolumn{3}{c}{$g_2(v)=v$}       &&  \multicolumn{3}{c}{$\Phi(v)$}      \\
                       \cline{4-6}\cline{8-10}
$g_1(u)=$       &RWS&& L1         &  L2        &  L3        &&  L1        &  L2         &   L3      \\
                      \hline
$u$    & 120 && -0.0473&   -0.0360&   -0.0535 &&  -1.0341&   -1.8261 &  -1.4626\\
                      & 240 && -0.0758&   -0.0451&   -0.0774 &&  -0.8457 &  -1.3144 &  -1.6810\\
                      & 360 && -0.0665&   -0.0378&   -0.0716 &&  -0.8198 &  -0.8714 &  -0.6289\\
$\Phi(u)$&
                       120 && \bf{0.0701} &   \bf{0.0490}  &  \bf{0.0669} &&  -1.0341 &  -1.8261 &  -1.4626\\
                       &240 && \bf{0.0970} &   \bf{0.0517}  &  \bf{0.1162} &&  -0.8457 &  -1.3144 &  -1.6810\\
                       &360 && \bf{0.1275} &   \bf{0.0530} &   \bf{0.1210} &&  -0.8198 &  -0.8714 &  -0.6289\\
\midrule
&&\multicolumn{8}{c}{Panel B: January 1952-December 2005}\\
\cline{2-10}
$u$    & 120 && -0.0306 &  -0.0263 &  -0.0335 &&  -0.2523 &  -0.1541 &  -0.2666\\
                      & 240 && -0.0580 &  -0.0351 &  -0.0579 &&  -0.7872 &  -0.4355 &  -0.7848\\
                      & 360 && -0.0748  & -0.0379 &  -0.0742 &&  -1.2610  & -0.6199 &  -1.2700\\
$\Phi(u)$&
                       120 && \bf{0.0466} &   \bf{0.0395}  &  \bf{0.0269} &&  -0.8947 &  -1.0879  & -1.2178\\
                       &240 && \bf{0.0664}  &  \bf{0.0286}  &  \bf{0.0844} &&  -0.4725 &  -0.6958 &  -0.6943\\
                       &360 && \bf{0.1026}  &  \bf{0.0319}  &  \bf{0.0950} &&  -0.7625 &  -0.5595 &  -0.4954\\
\bottomrule
\end{tabular}
\end{table}

{\footnotesize
\begin{sidewaystable}[htbp]
\centering
\caption{Pseudo Out-of-sample $R^2$ for quantile excess return prediction with $x=\{dp,tms\}$ and $z=\{inf, svar\}$.}\label{PR2_3}
\begin{tabular}[t]{lc l rrrrr r rrrrr}
\toprule
&&\multicolumn{12}{c}{Panel A: January 1927-December 2005}\\
\cline{2-14}
                       &        &&\multicolumn{11}{c}{$g_2(v)=v$}       \\
                       \cline{4-14}
$g_1(u)=$       &RWS&&$\tau=0.05$ &0.1&0.2&0.3&0.4&0.5&0.6&0.7&0.8&0.9&0.95\\
                      \hline
$u$    & 120 &&-0.0187 &   0.0032 &  -0.0109 &  -0.0099 &  -0.0241 &  -0.0360 &  -0.0408 & -0.0451 &  -0.0605 &  -0.0886 &  -0.0796\\
          & 240 &&0.0389  &  0.0133  &  0.0011  & -0.0108  & -0.0266  & -0.0451 &  -0.0469 &  -0.0677 &  -0.0921 &  -0.1123  & -0.1581\\
          & 360 &&0.0030  &  0.0305  &  0.0093 &  -0.0100 &  -0.0187 &  -0.0382 &  -0.0455  & -0.0714  & -0.0895 &  -0.1372 &  -0.1828\\
$\Phi(u)$&
                       120 && -0.0504  &  0.0165  &  0.0305  &  \bf{0.0439} &   0.0432 &   0.0490 &   \bf{0.0514}  & 0.0493  &  0.0497  &  0.0540  &  0.0338\\
                       &240 && 0.0569  &  0.0319  &  0.0237 &   0.0442 &   \bf{0.0649}  &  0.0517  &  \bf{0.0562} &  0.0461  &  0.0437 &   0.0471 &   0.0328\\
                       &360 && -0.0056 &   0.0039  &  0.0127 &   0.0542 &   0.0528 &   0.0530 &  0.0534 &  0.0381 &   0.0257 &   0.0268 &   0.0333\\

$Logit(u)$ & 120 && \bf{-0.0017} &   \bf{0.0251}  &  \bf{0.0318}  &  0.0393  &  \bf{0.0441}  &  \bf{0.0516}  &  0.0432  & \bf{0.0543}  &  \bf{0.0583}  &  \bf{0.0606}  &  \bf{0.0732}\\
          & 240 &&\bf{0.0667} &  \bf{ 0.0387}  &  \bf{0.0363} &   \bf{0.0536} &   0.0568 &   \bf{0.0564}  &  0.0538 &  \bf{0.0633} &   \bf{0.0534}  &  \bf{0.0629}  &  \bf{0.0764}\\
          & 360 &&\bf{0.0058} &   \bf{0.0131} &   \bf{0.0391} &   \bf{0.0560}  &  \bf{0.0540}  &  \bf{0.0620} &   \bf{0.0613} &  \bf{0.0648} &   \bf{0.0439}  &  \bf{0.0419} &   \bf{0.0714}\\
\midrule
&&\multicolumn{12}{c}{Panel B: January 1952-December 2005}\\
\cline{2-14}
$u$    & 120 &&-0.0169 &  -0.0011 &  -0.0080 &  -0.0091 &  -0.0105 &  -0.0262 &  -0.0251& -0.0316 &  -0.0426 &  -0.0552 &  -0.0629\\
          & 240 &&\bf{0.0459} &   0.0145  &  0.0064  & -0.0093  & -0.0149  & -0.0350 &  -0.0346 & -0.0518 &  -0.0727 &  -0.1068 &  -0.1245\\
          & 360 &&\bf{0.0406}  &  0.0320 &   0.0125  & -0.0147  & -0.0297 &  -0.0380 &  -0.0454 & -0.0492 &  -0.0709 &  -0.1351 &  -0.1811\\
$\Phi(u)$&
              120 && 0.0418  & -0.0034  &  0.0124  &  \bf{0.0392} &   0.0296  &  \bf{0.0395}  &  \bf{0.0454} & 0.0273 &   0.0228  &  \bf{0.0379}  &  0.0199\\
          &240 &&0.0357 &   0.0141 &   0.0130 &   0.0377 &   \bf{0.0481} &   0.0286  &  0.0293 &  0.0029 &  -0.0127  &  0.0033 &   0.0001\\
          &360 &&0.0027 &   0.0141 &   0.0092&    \bf{0.0707} &   0.0456 &   0.0319  &  0.0146  & -0.0268 &  -0.0933 &  -0.1060 &  -0.0830\\
$Logit(u)$ &
              120 &&\bf{0.0445}  &  \bf{0.0303}  &  \bf{0.0323} &   0.0326  &  \bf{0.0336} &   0.0385  &  0.0335&  \bf{0.0394} &   \bf{0.0273}  &  0.0378  &  \bf{0.0415} \\
          & 240 &&0.0452  &  \bf{0.0171} &   \bf{0.0204}   & \bf{0.0479}  &  0.0377   & \bf{0.0358} &   \bf{0.0306} &  \bf{0.0317} &   \bf{0.0029} &   \bf{0.0194}  &  \bf{0.0529}\\
          & 360 &&0.0140 &   \bf{0.0173} &   \bf{0.0395} &   0.0596  &  \bf{0.0474}  &  \bf{0.0415} &   \bf{0.0270} &  \bf{0.0143} &  \bf{-0.0612} &  \bf{-0.0746} &  \bf{-0.0116}\\
\bottomrule
\end{tabular}
\end{sidewaystable}
}

\section{Monte Carlo simulations}\label{simulation}

We consider the following data generating process:
\begin{equation}\label{dgp}
  y_t=g(x_{t}^\top {\bf\theta}, z_{t}^\top {\bf\beta})+e_t,
\end{equation}
where $x_t=\rho_{n1} x_{t-1}+\sigma_1 w_{t}$, $z_t=h(t/n)+v_t$, $h(\tau)=\tau$, $v_t=\rho_{n,2} v_{t-1}+\sigma_2 \epsilon_{t}$, both $x_t$ and $v_t$ are bivariate autoregressive vector processes with $x_0=0$,  $v_0=0$,  $\rho_{n1}=I_2$, $\rho_{n2}=0.5\times I_2$, $\sigma_1=diag\{0.2,0.5\}$, $\sigma_2=I_2$, and $(w _{t}^\top,\epsilon_{t}^\top)$ is a series of independent four dimensional normal random vector with zero mean and identity covariance matrix, $\theta=(1,0)^\top$, $\beta=(2,1)^\top$. The regression residual $e_t=0.1\cdot u_t$ with $u_t$ independently generated according to four distributions: (D1) standard normal;  (D2) mixed normal $0.9\cdot N(0,1)+0.1\cdot N(0,4)$; (D3) $t$ distribution with 2 degrees of freedom; (D4) standard Cauchy distribution.

For the bivariate nonlinear function $g(\cdot,\cdot)$, we consider the following designs:
\begin{itemize}
\item[(M1)]  $g(u,v)=u+v$,
\item[(M2)] $g(u,v)=\phi(u)+v$,
\item[(M3)] $g(u,v)=\Phi(u)+v$,
\item[(M4)] $g(u,v)=u+\phi(v)$,
\end{itemize}
where $\phi(\cdot)$ is the standard normal density function, $\Phi(\cdot)$ is the standard normal distribution function. Three loss functions are entertained: (L1) squared errors loss $\rho(e)=e^2$; (L2) absolute errors loss $\rho(e)=|e|$, and (L3) Huber's loss $\rho_\delta (e)=\frac{1}{2}e^2\cdot 1\{|e|\leq \delta\}+\delta\cdot (e-\frac{1}{2} \delta) \cdot 1\{|e|> \delta\}$ with $\delta=1.25$. The simulations are conducted for $n=100, 200, 400$ with 5,000 replications.

To measure the estimation accuracy of the nonlinear least-square estimates for the index parameters, we compute the bias, estimated standard deviations and mean squared errors for each element of the indices. For space limitation, we only report the mean squared errors in Tables \ref{rmse1}-\ref{rmse4} for M1-M4, respectively. Each table contains the estimation results under the three loss functions L1-L3, and four types of error distributions D1-D4.
The main findings are summarized as follows.  First, the mean squared errors are decreasing as the sample size $n$ increases, except when the errors are generated according to the Cauchy distribution and when the least squares loss is implemented. Under Cauchy errors, the least-square estimator is known to be inconsistent. However, the median estimator and the Huber's estimator remain consistent. As a result, we shall exclude this scenario when we further comment on the estimator performance below. Second, the index parameter estimate of the nonstationary variables enjoys super-rate of convergence, as shown in all the cases, as compared to the standard $\sqrt{n}$ rate for the stationary index estimator. Finally, the three estimates considered are quite competitive in terms of the different error distributions. It is found that, as expected, the least squares estimator is the most efficient estimator when the errors are drawn from normal distributions. The Huber's estimator becomes the most efficient, in general, for the mixed normal errors and $t(2)$ errors. The least absolute error estimate is the most efficient one when the errors are Cauchy.

{\samepage
{\footnotesize
\begin{sidewaystable}[htbp]
\caption{Mean Squared Error of M-estimator under M1 ($\times 10^3$). Numbers greater than $10^3$ are marked as ``*''.}
\label{rmse1}
\centering
\begin{tabular}{rcccccccccccccccc}
\toprule
&& \multicolumn{3}{c}{D1: Normal}&& \multicolumn{3}{c}{D2: Mixed normal}
&& \multicolumn{3}{c}{D3: t(2)}&& \multicolumn{3}{c}{D4: Cauchy}\\
\cline{3-5}\cline{7-9}\cline{11-13}\cline{15-17}
 &&L1&L2&L3&&L1&L2&L3&&L1&L2&L3&&L1&L2&L3\\
\midrule
$n=100$\\
$\theta_{1}$&&0.273&0.341&0.287&&0.357&0.392&0.334&&3.349&0.478&0.499&&*&0.596&0.806\\
$\theta_{2}$&&0.042&0.026&0.045&&0.052&0.029&0.051&&0.456&0.042&0.085&&*&0.046&0.132\\
$\beta_{1}$&&0.084&0.127&0.089&&0.101&0.134&0.095&&1.024&0.172&0.155&&*&0.230&0.245\\
$\beta_{2}$&&0.082&0.125&0.087&&0.105&0.136&0.100&&1.024&0.172&0.157&&*&0.224&0.242\\
$n=200$\\
$\theta_{1}$&&0.065&0.097&0.069&&0.085&0.107&0.080&&1.009&0.120&0.120&&*&0.166&0.198\\
$\theta_{2}$&&0.010&0.011&0.010&&0.013&0.012&0.012&&0.291&0.015&0.019&&*&0.020&0.031\\
$\beta_{1}$&&0.039&0.059&0.042&&0.050&0.068&0.047&&1.510&0.080&0.072&&*&0.100&0.110\\
$\beta_{2}$&&0.040&0.061&0.043&&0.051&0.066&0.049&&0.869&0.079&0.071&&*&0.101&0.110\\
$n=400$\\
$\theta_{1}$&&0.017&0.025&0.017&&0.021&0.028&0.021&&0.409&0.031&0.028&&*&0.040&0.045\\
$\theta_{2}$&&0.003&0.004&0.003&&0.003&0.004&0.003&&0.032&0.005&0.005&&*&0.006&0.007\\
$\beta_{1}$&&0.018&0.029&0.019&&0.025&0.032&0.023&&0.244&0.039&0.034&&*&0.047&0.053\\
$\beta_{2}$&&0.019&0.029&0.020&&0.025&0.034&0.024&&0.308&0.037&0.035&&*&0.049&0.053\\
\bottomrule
\end{tabular}
\end{sidewaystable}
}}

\begin{sidewaystable}[htbp]
\caption{Mean Squared Error of M-estimator under M2 ($\times 10^3$). Numbers greater than $10^3$ are marked as ``*''.}
\label{rmse2}
\centering
\begin{tabular}{rcccccccccccccccc}
\toprule
&& \multicolumn{3}{c}{D1: Normal}&& \multicolumn{3}{c}{D2: Mixed normal}
&& \multicolumn{3}{c}{D3: t(2)}&& \multicolumn{3}{c}{D4: Cauchy}\\
\cline{3-5}\cline{7-9}\cline{11-13}\cline{15-17}
 &&L1&L2&L3&&L1&L2&L3&&L1&L2&L3&&L1&L2&L3\\
\midrule
$n=100$\\
$\theta_{1}$&&1.545&1.075&1.598&&1.854&1.238&1.755&&*&1.511&2.899&&*&2.199&5.955\\
$\theta_{2}$&&0.352&0.026&0.349&&0.353&0.029&0.337&&*&0.029&0.614&&*&0.050&1.290\\
$\beta_{1}$&&0.008&0.012&0.008&&0.010&0.013&0.010&&0.089&0.016&0.015&&*&0.022&0.024\\
$\beta_{2}$&&0.008&0.012&0.008&&0.010&0.013&0.009&&0.100&0.016&0.014&&*&0.021&0.024\\
$n=200$\\
$\theta_{1}$&&0.650&0.502&0.675&&0.794&0.592&0.769&&*&0.770&1.366&&*&0.913&2.204\\
$\theta_{2}$&&0.144&0.018&0.159&&0.203&0.027&0.206&&*&0.032&0.262&&*&0.041&0.530\\
$\beta_{1}$&&0.004&0.006&0.004&&0.005&0.006&0.004&&0.035&0.007&0.007&&*&0.010&0.011\\
$\beta_{2}$&&0.004&0.006&0.004&&0.005&0.006&0.005&&0.039&0.007&0.007&&*&0.010&0.011\\
$n=400$\\
$\theta_{1}$&&0.342&0.356&0.367&&0.512&0.407&0.486&&*&0.437&0.751&&*&0.567&1.257\\
$\theta_{2}$&&0.071&0.013&0.079&&0.103&0.018&0.095&&*&0.019&0.159&&*&0.021&0.333\\
$\beta_{1}$&&0.002&0.003&0.002&&0.002&0.003&0.002&&0.020&0.003&0.003&&*&0.004&0.005\\
$\beta_{2}$&&0.002&0.003&0.002&&0.002&0.003&0.002&&0.023&0.003&0.003&&*&0.004&0.005\\
\bottomrule
\end{tabular}
\end{sidewaystable}

\begin{sidewaystable}[htbp]
\caption{Mean Squared Error of M-estimator under M3 ($\times 10^3$). Numbers greater than $10^3$ are marked as ``*''.}
\label{rmse3}
\centering
\begin{tabular}{rcccccccccccccccc}
\toprule
&& \multicolumn{3}{c}{D1: Normal}&& \multicolumn{3}{c}{D2: Mixed normal}
&& \multicolumn{3}{c}{D3: t(2)}&& \multicolumn{3}{c}{D4: Cauchy}\\
\cline{3-5}\cline{7-9}\cline{11-13}\cline{15-17}
 &&L1&L2&L3&&L1&L2&L3&&L1&L2&L3&&L1&L2&L3\\
\midrule
$n=100$\\
$\theta_{11}$&&9.219&9.533&10.15&&13.19&10.32&11.93&&*&13.01&24.69&&*&19.59&84.01\\
$\theta_{12}$&&2.058&0.202&1.881&&2.279&0.176&2.198&&*&0.244&5.514&&*&0.254&16.23\\
$\theta_{21}$&&0.077&0.116&0.082&&0.099&0.131&0.094&&0.866&0.161&0.149&&*&0.216&0.239\\
$\theta_{22}$&&0.078&0.116&0.082&&0.100&0.130&0.096&&1.094&0.163&0.150&&*&0.218&0.242\\
$n=200$\\
$\theta_{11}$&&5.331&5.673&5.777&&8.520&6.780&7.481&&*&8.102&14.89&&*&11.69&56.70\\
$\theta_{12}$&&0.859&0.260&0.896&&1.237&0.174&1.085&&*&0.193&3.578&&*&0.260&10.81\\
$\theta_{21}$&&0.035&0.054&0.037&&0.046&0.062&0.044&&0.432&0.072&0.065&&*&0.092&0.102\\
$\theta_{22}$&&0.036&0.056&0.038&&0.046&0.061&0.043&&0.516&0.069&0.065&&*&0.093&0.102\\
$n=400$\\
$\theta_{11}$&&3.896&4.364&4.307&&5.005&4.275&4.512&&*&5.642&7.141&&*&7.081&25.72\\
$\theta_{12}$&&0.641&0.181&0.666&&0.783&0.161&0.737&&*&0.204&1.093&&*&0.236&5.358\\
$\theta_{21}$&&0.016&0.025&0.018&&0.021&0.029&0.020&&0.229&0.035&0.031&&*&0.042&0.045\\
$\theta_{22}$&&0.017&0.027&0.018&&0.022&0.030&0.021&&0.246&0.034&0.031&&*&0.045&0.049\\
\bottomrule
\end{tabular}
\end{sidewaystable}

\begin{sidewaystable}[htbp]
\caption{Mean Squared Error of M-estimator under M4 ($\times 10^3$). Numbers greater than $10^3$ are marked as ``*''.}
\label{rmse4}
\centering
\begin{tabular}{rcccccccccccccccc}
\toprule
&& \multicolumn{3}{c}{D1: Normal}&& \multicolumn{3}{c}{D2: Mixed normal}
&& \multicolumn{3}{c}{D3: t(2)}&& \multicolumn{3}{c}{D4: Cauchy}\\
\cline{3-5}\cline{7-9}\cline{11-13}\cline{15-17}
 &&L1&L2&L3&&L1&L2&L3&&L1&L2&L3&&L1&L2&L3\\
\midrule
$n=100$\\
$\theta_{11}$&&0.274&0.369&0.288&&0.353&0.440&0.331&&2.860&0.531&0.537&&*&0.698&0.819\\
$\theta_{12}$&&0.041&0.044&0.044&&0.058&0.053&0.053&&0.519&0.063&0.084&&*&0.080&0.130\\
$\theta_{21}$&&31.78&49.76&34.24&&44.38&57.81&40.89&&*&66.72&64.26&&*&110.5&136.7\\
$\theta_{22}$&&11.27&17.31&12.19&&15.95&21.04&14.99&&*&23.92&23.12&&*&42.04&50.73\\
$n=200$\\
$\theta_{11}$&&0.067&0.104&0.071&&0.087&0.109&0.081&&0.814&0.130&0.120&&*&0.171&0.200\\
$\theta_{12}$&&0.011&0.016&0.012&&0.014&0.017&0.013&&0.113&0.019&0.019&&*&0.025&0.031\\
$\theta_{21}$&&14.12&22.44&15.06&&18.92&24.85&17.73&&*&31.88&28.39&&*&40.95&45.49\\
$\theta_{22}$&&5.102&8.171&5.454&&6.771&8.833&6.281&&*&10.91&9.927&&*&14.83&16.20\\
$n=400$\\
$\theta_{11}$&&0.016&0.026&0.017&&0.022&0.029&0.021&&1.091&0.033&0.030&&*&0.044&0.048\\
$\theta_{12}$&&0.003&0.004&0.003&&0.003&0.004&0.003&&0.361&0.005&0.005&&*&0.006&0.007\\
$\theta_{21}$&&6.725&10.680&7.168&&9.091&12.34&8.616&&*&14.36&13.23&&*&19.15&20.54\\
$\theta_{22}$&&2.381&3.798&2.524&&3.287&4.421&3.086&&*&5.247&4.680&&*&6.672&7.182\\
\bottomrule
\end{tabular}
\end{sidewaystable}

\section{Conclusion}\label{conclusion}

In order to cater for the practical usefulness this paper proposes a class of multiple index time series parametric models that accommodate time trend as well as both stationary and nonstationary vectors. An $M$-estimation approach is adopted where the loss functions can be, but not limited to, six popular ones, such as the squared loss, LAD, Huber's loss, quantile loss, $L_p$ and expectile loss. Meanwhile, two categories of link functions are investigated due to the different behaviour of functions of nonstationary vector variables, that is, $I$-regular and $H$-regular classes in the related literature. Accordingly, our asymptotic theory dwells on two categories of estimators and the rates of convergence are discussed under different classes of loss functions. Moreover, these models are used in the analysis of predictability where we find that our models are competitive comparing with the literature in terms of predictability. Finally, we conduct Monte Carlo simulations that reveal the satisfactory performance of our estimators proposed in finite sample situations.

\section*{Acknowledgements}

Dong would like to thank the financial support from National Natural Science Foundation of China (Grant 72073143); Gao acknowledges financial support from the Australian Research Council Discovery Grants Program under Grant Number: DP170104421; Peng acknowledges the Australian Research Council Discovery Grants Program for its financial support under Grant Number DP210100476; and Tu would like to thank support from National Natural Science Foundation of China (Grant 72073002, 12026607, 92046021), the Center for Statistical Science at Peking University, and Key Laboratory of Mathematical Economics and Quantitative Finance (Peking University), Ministry of Education.

{\footnotesize

\bibliography{123}

}

{\small

\section*{Appendix A: Preliminaries on generalized functions}

\renewcommand{\theequation}{A.\arabic{equation}} \renewcommand{\thelemma}{A.%
\arabic{lemma}} \renewcommand{\thesubsection}{A.\arabic{subsection}} %
\renewcommand{\thefigure}{A.\arabic{figure}}

\setcounter{equation}{0} \setcounter{lemma}{0} \setcounter{subsection}{0} %
\setcounter{figure}{0}

Since we shall adopt a generalized function approach in the proofs, some preliminaries about generalized functions are given in this section. In mathematic context generalized functions are called distributions or tempered distributions according to the spaces of test functions.

\medskip

\noindent {\bf Definition} (Space $D$) \ \emph{The space of all functions $\phi(x)$ defined on the real line satisfying}
\begin{itemize}
  \item \emph{$\phi(x)$ is an infinitely differentiable function defined at every point on $\mathbb{R}$. That is, $\phi^{(k)}(x)$ exists for any positive integer $k$};
  \item \emph{There is a constant $A>0$ such that $\phi(x)\equiv 0$ for $|x|>A$, or equivalently it has a compact support,}
\end{itemize}
\emph{is called Space $D$}.

Note that $D$ is a linear space over real set. A function $\phi(x)\in D$ is of $C^\infty$, also known as a test function. There exist many different types of functions in $D$, and it is noteworthy that for any continuous function $f(x)$ with compact support, there is a function $\phi(x)\in D$ such that $|f(x)-\phi(x)|< \varepsilon$ for all $x$ and any given $\varepsilon>0$. This implies the denseness of $D$ in any $C[a,b]$. See \citet[p. 3]{gelfand1964}.
\medskip

\noindent {\bf Definition} (Convergence in $D$) \ \emph{A sequence $\{\phi_m\}$ in $D$ is said to converge to a function $\phi_0$ if the following conditions are satisfied:}
\begin{itemize}
  \item \emph{All $\phi_m$ as well as $\phi_0$ vanish outside a common region};
  \item \emph{$\phi_m^{(k)}\to\phi_0^{(k)}$ uniformly over $\mathbb{R}$ as $m\to\infty$ for all $k\ge 0$}.
\end{itemize}

It can be shown that $\phi_0\in D$, and therefore $D$ is closed under this definition.

\medskip

\noindent {\bf Definition} (Distribution) \ \emph{A continuous linear functional on the space $D$ is called a distribution. The space of all distributions on $D$ is denoted by $D'$}.
\medskip

Distribution is another name of generalized function. For the definitions of continuity and linearity of a functional, please see \citet[p. 25]{kanwal1983}. The space $D'$ is called the dual space of $D$, is itself a linear space.

The set of distributions that are mostly useful are those generated by locally integrable functions. Indeed, if $f(x)$ is locally integrable on $\mathbb{R}$, it generates a distribution $f:\ \phi \mapsto \mathbb{R}$ through
\begin{equation*}
\langle f, \phi\rangle=\int f(x)\phi(x)dx, \ \ \ \forall \phi\in D.
\end{equation*}

Such defined distribution is called \emph{regular distribution}. Remarkably, it is proved in \citet[p. 27]{kanwal1983} that two continuous functions that generate the same regular distributions are identical. Moreover, if two locally integrable functions produce the same regular distributions, they are identical almost everywhere. These enable one to identify functions from the distributions they generate.

All distributions other than regular ones are called \emph{singular}. Thus, $D'$ is larger than $D$, because all $\phi\in D$ are distributions so that $D\subset D'$, and there do exist singular distributions, in particular Dirac delta $\delta(\cdot)$: $\phi\mapsto \phi(x_0)$, $\forall \, \phi\in D$, with a fixed $x_0\in \mathbb{R}$, is a singular distribution as shown in \citet[p. 4]{gelfand1964}.

\medskip

By definition, generalized functions cannot be assigned values at isolated points, while statements about a generalized function on a neighbourhood of points can be given in a well-defined way. This means that generalized functions are determined by its local property. Detailed discussions can be found in \citet[p. 5]{gelfand1964}. We should point out that it is because of this regard that we can circumvent the difficulty of non-smoothness of the loss functions in the $M$-estimation.

As is well known, not all ordinary functions are differentiable. In contradiction, generalized functions always have derivatives that are generalized functions too, and consequently have derivatives of any order. See \citet[p. 18]{gelfand1964}.

\medskip

\noindent {\bf Definition} (Derivative of distributions) \emph{Let $f\in D'$. A functional $g$ defined on $D$ given by}
\begin{equation*}
\langle g, \phi\rangle =-\langle f, \phi'\,\rangle, \ \ \ \forall\, \phi\in D
\end{equation*}
\emph{is called the derivative of $f$, denoted as $f'$ or $df/dx$}.
\medskip

The definition is motivated by integration by parts. To be an eligible member of $D'$, note that such defined $f'$ is linear and continuous. Accordingly, all generalized functions have derivatives of all orders. Moreover, the generalized derivative of a regular distribution agrees with the conventional one whenever the latter exists.
\medskip

\noindent {\bf Definition} (Convergence in $D'$) \ \emph{A sequence of distributions $f_m\in D'$, $m=1,2,\cdots$, is said to converge to a distribution $f\in D'$ if}
\begin{equation*}
\lim_{m\to\infty}\langle f_m,\phi\rangle=\langle f, \phi\rangle, \ \ \ \forall\, \phi\in D.
\end{equation*}
\emph{A set of distributions $\{f_\epsilon\}$ indexed by real $\epsilon$ is said converging to $f$ when $\epsilon\to\epsilon_0$, if for $ \forall\, \phi\in D$, $\lim_{\epsilon\to \epsilon_0} \langle f_\epsilon,\phi\rangle=\langle f, \phi\rangle$}.

\emph{A series of distributions $\sum_{m=1}^\infty f_m$ converges to a distribution $f\in D$ if the sequence of partial sum $s_M=\sum_{m=1}^M f_m$ converges to $f$ as $M\to\infty$.}

These definitions contain the convergence of ordinary functions as a special case. Indeed, suppose that all members of distribution sequence $\{f_m\}$ are regular, and $f_m(x)$ converge to $f(x)$ uniformly on any compact interval, then
\begin{equation*}
\lim_{m\to\infty}\langle f_m,\phi\rangle=\lim_{m\to\infty}\int f_m(x)\phi(x)dx=\int f(x)\phi(x)dx=\langle f, \phi\rangle, \ \ \forall\, \phi\in D,
\end{equation*}
by uniform convergence theorem.

A consequence of the definition is that if $f_m\to f$ then $f_m^{(k)}\to f^{(k)}$ for any $k>0$; if $\sum_{m=1}^\infty f_m$ converges to $f$, then the series can be differentiated term by term as many times as required. See \citet[p. 59]{kanwal1983}.

The most important sequence of distributions is a sequence of regular distributions $\{f_m\}$, so-called delta-convergent sequence, that converges to $\delta$ distribution. This may be regarded as a bridge between regular and singular distributions.

\begin{lemma}
Let $f(x)$ be a nonnegative function defined on $\mathbb{R}$ such that $\int f(x)dx=1$. Put $f_\epsilon(x)=\epsilon^{-1}f(x/\epsilon)$, $\epsilon>0$. Then $\lim_{\epsilon\to 0}f_\epsilon(x)=\delta(x)$.
\end{lemma}

This is exactly the theorem in \citet[p. 62]{kanwal1983} for univariate functions. This result enables us to construct delta sequence. Observe that $\epsilon$ can be replaced by $1/m$ to have sequence $f_m(x)$ that convergence to $\delta(x)$ as $m\to\infty$.

{\bf Example} (Delta-convergent sequence)\

(1) $f(x)=I(|x|\le 1/2)$, $f_m(x)=mI(|mx|\le 1/2)$. Then $\lim_{m\to \infty}f_m(x)=\delta(x)$.

(2) $f(x)=\frac{1}{\pi}\frac{1}{x^2+1}$, and $f_\epsilon(x)=\frac{1}{\pi} \frac{\epsilon}{x^2+\epsilon^2}$. Then $\lim_{\epsilon\to 0}f_\epsilon(x) =\delta(x)$.

(3) Define $f_\epsilon(x)=\frac{1}{\sqrt{2\pi}\epsilon} \exp(-x^2/2\epsilon^2)$ with $\epsilon>0$. Then $\lim_{\epsilon\to 0} f_\epsilon(x)=\delta(x)$, and hence for any $g(x)$ that is continuous at $x=x_0$, we have
\begin{equation*}
\lim_{\epsilon\to 0}\int f_\epsilon(x-x_0)g(x)dx=\int \delta(x-x_0) g(x)dx =g(x_0).
\end{equation*}
This actually is the rationale behind the kernel estimation.

(4) Consider
\begin{equation*}
  f(x)=\begin{cases}
  C\exp\left(-\frac{1}{1-x^2}\right), & |x|<1,\\
  0, & |x|\ge 1,
  \end{cases}
\end{equation*}
where $C$ is such that $\int f(x)dx=1$. Then, we define
\begin{equation*}
  f_\epsilon(x)=\begin{cases}
  C\epsilon^{-1}\exp\left(-\frac{\epsilon^2}{\epsilon^2-x^2}\right), & |x|<\epsilon,\\
  0, & |x|\ge \epsilon,
  \end{cases}
\end{equation*}
and $\lim_{\epsilon\to 0}f_\epsilon(x)=\delta(x)$.
\medskip

In some cases we should consider a space that is larger than $D$ as test function space, in order to extend the compact support of test functions in $D$ to the entire real line.

\medskip

\noindent {\bf Definition} (Space $S$)\ The space $S$ of test functions of rapid decay contains all functions $\phi$ defined on $\mathbb{R}$ that satisfy
\begin{itemize}
  \item $\phi(x)$ is infinitely differentiable, i.e. $\phi(x)\in C^\infty$;
  \item $\phi(x)$, as well as its derivatives of all orders, vanishes at infinity faster than any power of $1/|x|$, i.e. for any $p,k\ge 0$, $|x^p\phi^{(k)}(x)|\le C_{pk}$ where the constant $C_{pk}$ only depends on $p$, $k$ and $\phi$.
\end{itemize}

The space $S$ is linear and clearly $D\subset S$. Accordingly $S'\subset D'$ because a continuous linear functional on $S$ is also a continuous linear functional on $D$. Similarly to $D$ and $D'$, we may define the convergence of sequence and the derivative of distributions in $S$ and $S'$. One may find these in \citet[p.17]{gelfand1964} and \citet[p.138]{kanwal1983}.

\section*{Appendix B: Lemmas}

\renewcommand{\theequation}{B.\arabic{equation}} \renewcommand{\thelemma}{B.%
\arabic{lemma}} \renewcommand{\thesubsection}{B.\arabic{subsection}} %
\renewcommand{\thefigure}{B.\arabic{figure}}

\setcounter{equation}{0} \setcounter{lemma}{0} \setcounter{subsection}{0} %
\setcounter{figure}{0}

Under Assumption A and w.l.o.g. letting $x_0=0$ a.s., similar to (A.5) and (A.6) of \citet{dgdy2017}, we have
\begin{align}\label{xt}
x_t=\sum_{i=1}^tw_i=\sum_{j=-\infty}^tB_{t,j}\eta_{j}, \ \ \text{where}\ B_{t,j}=\sum_{i=\max(1,j)}^tA_{i-j},
\end{align}
and for $t>s>0$,
\begin{align}\label{xts}
x_t=&\sum_{i=1}^tw_i=\sum_{i=s+1}^tw_i+x_s=x_{ts}+x_{ts}^*,
\end{align}
where $x_{ts}=\sum_{j=s+1}^tB_{t,j}\eta_{j}$, $x_{ts}^*=x_s+ \bar{x}_{ts}$,  and $\bar{x}_{ts}=\sum_{j=-\infty}^s \left(\sum_{i=s+1}^tA_{i-j}\right)\eta_{j}$. As a result, $x_{ts}$ is independent of $x_{ts}^*$. Denote $d_{ts}^2:= \mathbb{E}(x_{ts} x_{ts}^\top)$ and from the B-N decomposition (\citealp{phillips1992}) we have $d_{ts}^2\sim t-s$ when $t-s$ is large. The representations of \eqref{xt} and \eqref{xts}, along with the following lemma, will be used for asymptotic analysis.

\begin{lemma}\label{lemma0}
Let Assumption A hold. Let $x_{0,t}=\theta_0^\top x_t$ be unit root process and define $x_{0,ts}=\theta_0^\top x_{ts}$ for $t>s$ where $x_{ts}$ is given by \eqref{xts}; and define $d_{0,t}^2=\mathbb{E}(x_{0,t}^2)$ and $d_{0,ts}^2=\mathbb{E}(x_{0,ts}^2)$.
\begin{enumerate}[(1)]
\item For large $t$, $d_{0,t}^{-1}x_{0,t}$ have densities $f_t(u)$ which are uniformly bounded over $u\in \mathbb{R}$ and $t$. Meanwhile, the partial derivatives of $f_t(u)$ are uniformly bounded as well. Consequently, $f_t(u)$ satisfy a uniform Lipschitz condition
\begin{equation}\label{lemma01}
\sup_{u\in \mathbb{R}}|f_t(u+\triangle u)-f_t(u)|\le C \,|\triangle u|
\end{equation}
for some $C>0$ and any $\triangle u\in \mathbb{R}$.

\item For large $t-s$, $d_{0,ts}^{-1}x_{0,ts}$ have uniformly bounded densities $f_{ts}(u)$ over all $u\in \mathbb{R}$ and $(t,s)$. Additionally, $f_{ts}(u)$ have bounded partial derivatives and satisfy Lipshitz condition similar to \eqref{lemma01} as well.
\end{enumerate}
\end{lemma}

This is Lemma A.2 in \citet{donggao2018}.

Let $\theta$ be a unit vector in $\mathbb{R}^{d_1}$ such that the sequence $\theta^\top x_t$ is of unit root, and $R=(\theta, R_2)$ be an orthogonal matrix of dimension $d_1\times d_1$. Denote, if no ambiguity arises, $R^\top x_t=(x_{1t}, x_{2t}^\top)^\top$ conformably with block representation of $R$. It is clear that the covariance matrix of $R^\top x_t$ has asymptotics $R^\top A A^\top R\, t (1+o(1))$. We suppress $\theta$ from $x_{1t}$ and $x_{2t}$ for simplicity.

\begin{lemma}\label{lemma1}
Let Assumption A hold.
\begin{enumerate}[(1)]
\item For large $t$, $\sqrt{t}^{\;-1}(x_{1t},x_{2t}^\top)^\top$ have densities $\psi_{t}(x, w)$ which are uniformly bounded over $x\in \mathbb{R}, w\in \mathbb{R}^{d_1-1}$ and $t$. Meanwhile, the derivatives of $\psi_{t}(x,w)$ are uniformly bounded as well. Consequently, $\psi_{t}(x,w)$ satisfy a uniform Lipschitz condition
\begin{equation}\label{lemma11}
\sup_{x, w}|\psi_{t}(x+\triangle x, w+\triangle w)-\psi_{t}(x, w)|\le C( |\triangle x|+\|\triangle w\|)
\end{equation}
for some $C>0$ and any $\triangle x$ and $\triangle w$.

Meanwhile, for large $t$, $\psi_{t}(x,w)=f_{1t}(x)\phi(w)(1+o(1))$ where $f_{1t}(x)$ is the marginal density of $\sqrt{t}^{\;-1}x_{1t}$ and $\phi(w)$ is the standard normal density.

\item For large $t-s$, $d_{\theta,ts}^{-1}x_{\theta,ts}$ have uniformly bounded densities $f_{\theta,ts}(x)$ over all $x\in \mathbb{R}$ and $t>s$, where $x_{\theta,ts}=\theta^\top x_{ts}$ and $d_{\theta,ts}^{2}=\mathbb{E}x_{\theta,ts}^2$; additionally, $f_{\theta,ts}(x)$ have bounded partial derivatives and satisfy Lipshitz condition similar to \eqref{lemma11} as well.
\end{enumerate}

\end{lemma}

\begin{lemma}\label{lemma5a}
The following assertions hold:
\begin{enumerate}[(1)]
  \item $\frac{1}{\sqrt{t}}(x_{1t}, x_{2t}^\top)$ has a joint probability density $\psi_t(x, w^\top)$; and given $\mathcal{F}_{s}$ (defined in Assumption A), $\frac{1}{ \sqrt{t-s}} (x_{1t}-x_{1s}, x_{2t}^\top-x_{2s}^\top)$ has a joint density $\psi_{ts}(x, w^\top)$ where $t>s+1$. Meanwhile, these functions are bounded uniformly in $(x,w)$ as well as $t$ and $(t,s)$, respectively.

  \item For large $t$ and $t-s$, we have $\psi_t(x, w^\top)=\phi(w)f_t(x)(1+o(1))$ and $\psi_{ts}(x, w^\top)= \phi(w)f_{ts}(x) (1+o(1))$  where $\phi(w)$ is the density of a multivariate $(d-1)$-dimensional normal distribution, $f_t(x)$ is the marginal density of $\frac{1}{\sqrt{t}}x_{1t}$ and $f_{ts}(x)$ is the marginal density of $\frac{1}{\sqrt{t-s}}(x_{1t}-x_{1s})$.
\end{enumerate}
\end{lemma}

\begin{lemma}\label{lemma2}
Suppose that $\rho'(u)$ and $\rho''(u)$ exist in the ordinary sense. Under Assumptions A-C and D(1)-(2), denoting $x_{nt}=n^{-1/2}x_t$, as $n\to\infty$,
\begin{align}\label{lemma2a}
&\frac{1}{\sqrt{n}}\sum_{t=1}^n\rho'(e_t)\begin{pmatrix}
\dot{g}_1(\theta_0^\top x_{nt})x_{nt}\\
\dot{g}_2(\beta_0^\top z_t)z_t
\end{pmatrix}
\to_D
\begin{pmatrix}
\int_0^1 \dot{g}_1(\theta_0^\top B(r))]B(r) dU(r)\\
N(0, a_1\Sigma),
\end{pmatrix}
\end{align}
where $\Sigma=\int_0^1[\mathbb{E}\dot{g}_2^2 (\beta_0^\top h(r,v_1) ) \, h(r,v_1) h(r,v_1)^\top ]dr$ and $a_1>0$ given in Assumption C.

Moreover,
\begin{align}
&\frac{1}{n}\sum_{t=1}^n\rho''(e_t)g_2^2(\beta_0^\top z_t)z_tz_t^\top\to_P
a_2\Sigma, \label{lemma2b}\\ \intertext{and}
&\frac{1}{n}\sum_{t=1}^n\rho''(e_t)\begin{pmatrix}
\dot{g}_1^2(\theta_0^\top x_{nt})x_{nt}x_{nt}^\top\\
\dot{g}_1(\theta_0^\top x_{nt} )\dot{g}_2( \beta_0^\top z_t)z_tx_{nt}^\top
\end{pmatrix}\notag \\
&\to_D
\begin{pmatrix}
a_2\int_0^1 \dot{g}_1^2(\theta_0^\top B(r))B(r)B(r)^\top dr\\
a_2\int_0^1\dot{g}_1(\theta_0^\top B(r))[\mathbb{E}\dot{g}_2( \beta_0^\top h(r,v_1))h(r,v_1)]B(r)^\top dr
\end{pmatrix}. \label{lemma2c}
\end{align}
where $a_2=\mathbb{E}[\rho''(e_t)|\ra_{t-1}]>0$ a.s. stipulated in Assumption C.
\end{lemma}

\begin{proof}[Proof of Lemma \ref{lemma2}]

Recall that $z_t=h(\tau_t,v_t)$, and $\tau_t=t/n$ and $v_t=q(\eta_t,\cdots, \eta_{t-d_0+1})+\tilde{v_t}$, where $\tilde{v}_t$ is independent of $\{\eta_j, j\in \mathbb{Z}\}$. Thus, $z_t$ and $x_t$ are correlated through these $\eta$'s. In view of \eqref{xt} we may write $x_t=\sum_{i=t-d_0+1}^t w_i+ \sum_{i=1}^{t-d_0}w_i\equiv w_{t,d_0}+x_{t-d_0}$. Then, because $n^{-1/2}w_{t,d_0}=o_P(1)$, $x_{t}$ and $x_{t-d_0}$ have the same asymptotic distribution. Meanwhile, $x_{t-d_0}$ and $z_t$ are mutually independent.

By the continuity of $\ddot{g}_1(\cdot)$, $\dot{g}_1(\theta_0^\top x_{nt})- \dot{g}_1(\theta_0^\top x_{n,t-d_0})=O_P(n^{-1/2})$. We therefore may replace $\dot{g}_1(\theta_0^\top x_{nt})$ by $\dot{g}_1(\theta_0^\top x_{n, t-d_0})$ in the derivation in the sequel but we avoid doing so and treat them independent for simplicity.

It follows from \citet{phillips2001} that
\begin{equation*}
\frac{1}{\sqrt{n}}\sum_{t=1}^n\rho'(e_t)\dot{g}_1 (\theta_0^\top x_{n,t}) x_{n,t}\to_D \int_0^1 \dot{g}_1(\theta_0^\top B(r))B(r) dU(r),
\end{equation*}
while since $z_t$ is a strictly stationary and $\alpha$-mixing stationary, and due to the martingale difference structure imposed in Assumption B it is evidently that $\frac{1}{\sqrt{n}}\sum_{t=1}^n\rho'(e_t)\dot{g}_2(\beta_0^\top z_t)z_t\to_D N(0,a_1\Sigma)$,
where the conditional covariance matrix converging to $\Sigma$ in probability is shown below. The joint convergence then follows by the independence between $x_{n,t-d_0}$ and $z_t$. This finishes \eqref{lemma2a}.

Now, consider the convergence in probability in \eqref{lemma2b}. Because of the martingale difference structure for $e_t$, the result will hold if we can shown $\frac{1}{n} \sum_{t=1}^n\dot{g}_2^2(\beta_0^\top z_t)z_tz_t^\top\to_P \Sigma$ as $n\to\infty$. To this end, notice that $z_t$ is $\alpha$-mixing, so that by Assumption B $\frac{1}{n} \sum_{t=1}^n[\dot{g}_2^2(\beta_0^\top z_t)z_tz_t^\top- \mathbb{E}\dot{g}_2^2(\beta_0^\top z_t)z_tz_t^\top]=o_P(1)$. Indeed,
\begin{align*}
&\mathbb{E}\left\| \frac{1}{n} \sum_{t=1}^n[\dot{g}_2^2(\beta_0^\top z_t)z_tz_t^\top- \mathbb{E}\dot{g}_2^2(\beta_0^\top z_t)z_tz_t^\top]\right\|^2 = \frac{1}{n^2}\mathbb{E}\sum_{t=1}^n\left\|\dot{g}_2^2(\beta_0^\top z_t)z_tz_t^\top- \mathbb{E}[\dot{g}_2^2(\beta_0^\top z_t)z_tz_t^\top]\right\|^2\\
&+\frac{2}{n^2}\mathbb{E}\sum_{t=2}^n\sum_{s=1}^{t-1}\text{tr}\left([ \dot{g}_2^2(\beta_0^\top z_t)z_tz_t^\top- \mathbb{E}\dot{g}_2^2(\beta_0^\top z_t)z_tz_t^\top] [ \dot{g}_2^2(\beta_0^\top z_s)z_sz_s^\top- \mathbb{E}\dot{g}_2^2(\beta_0^\top z_s) z_sz_s^\top]^\top\right)\\
\le&\frac{1}{n^2}\sum_{t=1}^n\mathbb{E}[\dot{g}_2^4(\beta_0^\top z_t)\|z_t\|^4] +\frac{2}{n^2}\sum_{t=2}^n\sum_{s=1}^{t-1}\alpha^{1/2}(t-s)
[\mathbb{E}(\|\dot{g}_2^2(\beta_0^\top z_t)\;z_tz_t^\top\|^4)]^{1/4}[\mathbb{E}(\|\dot{g}_2^2(\beta_0^\top z_s)\;z_sz_s^\top\|^4)]^{1/4}\\
=&O(n^{-1})=o(1)
\end{align*}
where we use the Davydov's inequality for $\alpha$-mixing processes with $p=2$ and $q=r=4$ \citep[p 19]{bosq1996} and the condition in Assumption B(2).

It then suffices to show that
\begin{align*}
\frac{1}{n} \sum_{t=1}^n\mathbb{E}\dot{g}_2^2(\beta_0^\top z_t)z_tz_t^\top =&\frac{1}{n} \sum_{t=1}^n\mathbb{E}\dot{g}_2^2(\beta_0^\top h(\tau_t,v_t))h(\tau_t,v_t) h(\tau_t,v_t)^\top\\
=&\sum_{t=1}^{n-1}\int_{\tau_{t}}^{\tau_{t+1}}
\mathbb{E}\dot{g}_2^2(\beta_0^\top h(r,v_t)
h(r,v_t) h(r,v_t)^\top dr+o(1)\\
=&\int_0^1\mathbb{E}\dot{g}_2^2(\beta_0^\top(h(r)+v_1))
(h(r)+v_1)) (h(r)+v_1))^\top dr+o(1) \to \Sigma.
\end{align*}

Finally we consider \eqref{lemma2c}. As argued above, it suffices to consider the convergence of
\begin{align*}
&\frac{1}{n}\begin{pmatrix}
\sum_{t=1}^n\dot{g}_1^2(\theta_0^\top x_{nt})x_{nt}^\top\\
\sum_{t=1}^n\dot{g}_1(\theta_0^\top x_{nt})(\mathbb{E}[\dot{g}_2(\beta_0^\top z_t)z_t])x_{nt}^\top
\end{pmatrix}\\
=&\frac{1}{n}\begin{pmatrix}
\sum_{t=1}^n\dot{g}_1^2(\theta_0^\top x_{nt})x_{nt}^\top\\
\sum_{t=1}^n\dot{g}_1(\theta_0^\top x_{nt})(\mathbb{E}[\dot{g}_2(\beta_0^\top
h(\tau_t,v_1))h(\tau_t,v_1)])x_{nt}^\top
\end{pmatrix}.
\end{align*}
The result then follows from Theorem 3.2 of \citet[p. 131]{dg2019}.
\end{proof}

Consider the case where $\rho'(\cdot)$ or $\rho''(\cdot)$ is a generalized function. Let $\rho'_m(\cdot)$ and $\rho''_m(\cdot)$ be regular sequences of $\rho'(\cdot)$ and $\rho''(\cdot)$, respectively. See \citet[p 918]{phillips1995}. For Dirac delta function, for example, the regular sequence is any delta-convergent sequence. Define, for $r\in [0,1]$,
\begin{equation*}
U_{mn}(r)\equiv\frac{1}{\sqrt{n}}\sum_{t=1}^{[nr]}[\rho'_m(e_t)
-\mathbb{E}\rho'_m(e_t)].
\end{equation*}
Then for each $m$, by the  functional invariant principle, $U_{mn}(r)\Rightarrow U_{m}(r)$ as $n\to\infty$. Here, $U_{m}(r)$ is a Brownian motion with variance $\text{Var}(\rho'_m(e_1))$.

\begin{lemma}\label{lemma3}
Under Assumptions A-C and D(1)-(2), for each $m$, as $n\to\infty$,
\begin{align}\label{lemma3a}
\frac{1}{\sqrt{n}}\sum_{t=1}^n[\rho'_m(e_t)-\mathbb{E}\rho'_m(e_t)]
\begin{pmatrix}
\dot{g}_1(\theta_0^\top x_{nt})x_{nt}\\
g_2(\beta_0^\top z_t)z_t
\end{pmatrix}\to_D
\begin{pmatrix}
\int_0^1 \dot{g}_1(\theta_0^\top B(r))B(r) dU_{m}(r)\\
N(0, a_{m1}\Sigma)
\end{pmatrix},
\end{align}
where $\Sigma$ is the same as in Lemma \ref{lemma2}, $a_{m1}= \text{Var}(\rho'_m(e_1))$.

In addition, as $n\to\infty$,
\begin{align}
&\frac{1}{n}\sum_{t=1}^n\rho''_m(e_t)g_2^2(\beta_0^\top z_t)z_tz_t^\top\to_P
a_{m2}\Sigma, \label{lemma3b}\\ \intertext{and}
&\frac{1}{n}\sum_{t=1}^n\rho''_m(e_t)\begin{pmatrix}
\dot{g}_1^2(\theta_0^\top x_{nt})x_{nt}x_{nt}^\top\\
\dot{g}_1(\theta_0^\top x_{nt})g_2(\beta_0^\top z_t)z_tx_{nt}^\top
\end{pmatrix}\notag \\
& \hspace{1cm}\to_D
\begin{pmatrix}
a_{m2}\int_0^1 \dot{g}_1^2(\theta_0^\top B(r))B(r)B(r)^\top dr\\
a_{m2}\int_0^1\dot{g}_1(\theta_0^\top B(r))[\mathbb{E} g_2(\beta_0^\top h(r,v_1)) h(r,v_1)]B(r)^\top dr
\end{pmatrix},
\label{lemma3c}
\end{align}
where $a_{m2}=\mathbb{E}[\rho''_m(e_t)]>0$.
\end{lemma}

\begin{proof}[Proof of Lemma \ref{lemma3}]
The proof is the same as that of Lemma \ref{lemma2}.
\end{proof}

Denote
\begin{align*}
\xi_{nt}:=\begin{pmatrix} \frac{1}{\sqrt[4]{n}}\dot{g}_1(x_{1t})x_{1t}\\
 \frac{1}{\sqrt[4]{n}^3} \dot{g}_1(x_{1t})x_{2t}\\
 \frac{1}{\sqrt{n}}\dot{g}_2( z_t^\top\beta_0) z_t
\end{pmatrix} \ \ \ \text{and}\ \ \
M_{n,t}:=\xi_{nt}\xi_{nt}^\top.
\end{align*}

\begin{lemma}\label{lemma4}
Suppose that $\rho'(u)$ and $\rho''(u)$ exist in ordinary sense. Under Assumptions A-C and D(2)-(3), as $n\to\infty$ we have jointly
\begin{align}\label{lemma4a}
 \sum_{t=1}^n\rho''(e_t)M_{n,t}\to_D a_2\,M,
\end{align}
where $a_2>0$ is given in Assumption C and $M$ is a $d\times d$ symmetric matrix
\begin{align}\label{lemma4b}
M=\begin{pmatrix}\int [u\dot{g}_1(u)]^2du L_1(1,0) & \int_0^1 B_2^\top(r)dL_1(r, 0) \int u \dot{g}_1^2(u)du & 0\\
& \int_0^1B_2(r)B_2(r)^\top dL_1(r,0)\int \dot{g}_1^2(u)du & 0\\
& & \Sigma
\end{pmatrix},
\end{align}
where $\Sigma$ is defined in Lemma \ref{lemma2}.

Meanwhile,
\begin{align}\label{lemma4c}
 \sum_{t=1}^n \rho'(e_t)\xi_{nt}
\to_D M^{1/2}\xi \ \ \ \mbox{with} \ \ \ \xi\sim N(0, a_1I_d),
\end{align}
where $a_1>0$ is defined in Assumption C.
\end{lemma}

\begin{proof}[Proof of Lemma \ref{lemma4}]
We first show \eqref{lemma4a}. The joint convergence is based on the assumption of joint convergence for the underlying process in Assumption C.

Notice that $\sum_{t=1}^n\rho''(e_t)M_{n,t}=a_2\sum_{t=1}^nM_{n,t}+
\sum_{t=1}^n[\rho''(e_t)-a_2]M_{n,t}=a_2\sum_{t=1}^nM_{n,t}+o_P(1)$ by the martingale difference sequence structure in Assumption C. What is needed to show \eqref{lemma4a} is to establish the limit of $\sum_{t=1}^nM_{n,t}$.

Observe that as $n\to\infty$,
\begin{align*}
\frac{1}{\sqrt{n}}\sum_{t=1}^n[\dot{g}_1(x_{1t})x_{1t}]^2 &\to_D \int [u\dot{g}_1(u)]^2du\, L_1(1,0), \\
\frac{1}{n}\sum_{t=1}^n \dot{g}_1(x_{1t})^2x_{1t}x_{2t} &\to_D \int_0^1 B_2(r)dL_1(r, 0) \int u \dot{g}_1^2(u)du,\\
\frac{1}{n\sqrt{n}}\sum_{t=1}^n \dot{g}_1^2(x_{1t})x_{2t}x_{2t}^\top &\to_D \int_0^1B_2(r)B_2(r)^\top dL_1(r,0)\int \dot{g}_1^2(u)du
\end{align*}
were shown in \citet{phillips2000}, whereas $\frac{1}{n}\sum_{t=1}^n \dot{g}_2^2( z_t^\top\beta_0) z_tz_t^\top\to_P\Sigma$ shown in Lemma \ref{lemma2}. We now turn to the convergence of
\begin{equation*}
(1)\ \ \frac{1}{\sqrt[4]{n}^3}\sum_{t=1}^n \dot{g}_1(x_{1t})x_{1t}\dot{g}_2( z_t^\top\beta_0) z_t=o_P(1), \ \ \text{and}\ \ (2)\ \ \frac{1}{n\sqrt[4]{n}} \sum_{t=1}^n \dot{g}_1^2(x_{1t})x_{2t}\dot{g}_2( z_t^\top\beta_0)z_t^\top=o_P(1).
\end{equation*}

(1) Write
\begin{align*}
 &\frac{1}{\sqrt[4]{n}^3}\sum_{t=1}^n \dot{g}_1(x_{1t})x_{1t}\dot{g}_2( z_t^\top\beta_0) z_t\\
 =&\frac{1}{\sqrt[4]{n}^3}\sum_{t=1}^n \dot{g}_1(x_{1t})x_{1t}\mathbb{E}\dot{g}_2( z_t^\top\beta_0) z_t+ \frac{1}{\sqrt[4]{n}^3}\sum_{t=1}^n \dot{g}_1(x_{1t})x_{1t}\left[\dot{g}_2( z_t^\top\beta_0) z_t- \mathbb{E}\dot{g}_2( z_t^\top\beta_0) z_t\right]\\
=&A_{1n}+A_{2n}, \ \ \ \text{say}.
\end{align*}

First, $A_{2n}=o_P(1)$. Indeed, as argued in the proof of Lemma \ref{lemma2}, $x_t$ and $z_t$ can be regarded as mutually independent since they only share a finite number of $\eta$'s in the assumption, and we thus have
\begin{align*}
&\mathbb{E}\|A_{2n}\|^2=\frac{1}{\sqrt{n}^3}\mathbb{E}\sum_{t=1}^n \dot{g}_1^2(x_{1t})x_{1t}^2\left\|\dot{g}_2( z_t^\top\beta_0) z_t- \mathbb{E}\dot{g}_2( z_t^\top\beta_0) z_t\right\|^2  \\
   &+\frac{2}{\sqrt{n}^3}\mathbb{E}\sum_{t=2}^n \sum_{s=1}^{t-1} \dot{g}_1(x_{1t})x_{1t}\dot{g}_1(x_{1s})x_{1s}\\
   &\times \left[\dot{g}_2( z_t^\top\beta_0) z_t- \mathbb{E}\dot{g}_2( z_t^\top\beta_0) z_t\right]'\left[\dot{g}_2( z_s^\top\beta_0) z_s- \mathbb{E}\dot{g}_2( z_s^\top\beta_0) z_s\right]\\
\le&\frac{1}{\sqrt{n}^3}\sum_{t=1}^n \mathbb{E}[ \dot{g}_1^2(x_{1t})x_{1t}^2] \mathbb{E}\left\|\dot{g}_2( z_t^\top\beta_0) z_t- \mathbb{E}\dot{g}_2( z_t^\top\beta_0) z_t\right\|^2  \\
&+\frac{2}{\sqrt{n}^3}\sum_{t=2}^n \sum_{s=1}^{t-1} \mathbb{E}|\dot{g}_1(x_{1t})x_{1t}\dot{g}_1(x_{1s})x_{1s}| \alpha^{1/2}(t-s) \\
&\times \mathbb{E}\left\|\dot{g}_2( z_t^\top\beta_0) z_t- \mathbb{E}\dot{g}_2( z_t^\top\beta_0) z_t\right\|^4\mathbb{E}\left\|\dot{g}_2( z_s^\top\beta_0) z_s- \mathbb{E}\dot{g}_2( z_s^\top\beta_0) z_s\right\|^4 = O(n^{-1/2}),
\end{align*}
since $\sum_k\alpha^{1/2}(k)<\infty$, $\mathbb{E}\left\|\dot{g}_2( z_t^\top\beta_0) z_t- \mathbb{E}\dot{g}_2( z_t^\top\beta_0) z_t\right\|^4$ is bounded, $x_{1t}$ is unit root, $\dot{g}_1^2(u)u^2$ is integrable and $(x_{1t}, x_{1s})$ for $t-s$ large has joint density given $\ra_s$ that satisfies Lemma \ref{lemma0}. The detailed calculation is referred to Lemma A.3 in \citet{donggao2018}.

Second, $A_{1n}=o_P(1)$. This is shown by invoking the density of $d_{1t}^{-1}x_{1t}$ in Lemma \ref{lemma0} and the integrability of $\dot{g}_1^2(u)u^2$ as well as the boundedness of $\mathbb{E}\left\|\dot{g}_2( z_t^\top\beta_0) z_t- \mathbb{E}\dot{g}_2( z_t^\top\beta_0) z_t\right\|^2$. This finishes the proof of (1).

The assertion (2) can be shown similarly but one has to invoke Lemma \ref{lemma1} where the joint densities of $(x_{1t},x_{2t})$ and their properties have established. We omit the details due to the similarity.

Now we consider \eqref{lemma4c}. It is a martingale sequence and we thus invoke the martingale CLT, corollary 3.1 of \citet[p. 58]{peterhall1980}. Under Assumption C, we may show
\begin{equation*}
M_n^{-1/2}\sum_{t=1}^n \rho'(e_t)\xi_{nt}
\to_D \xi,\ \ \ \text{where}\ \xi_{nt}:=\begin{pmatrix} \frac{1}{\sqrt[4]{n}}\dot{g}_1(x_{1t})x_{1t}\\
 \frac{1}{\sqrt[4]{n}^3} \dot{g}_1(x_{1t})x_{2t}\\
 \frac{1}{\sqrt{n}}\dot{g}_2( z_t^\top\beta_0) z_t
\end{pmatrix} \ \ \mbox{and} \ \ M_n:=\sum_{t=1}^nM_{n,t}.
\end{equation*}
In fact, since the conditional variance is exactly $a_1I_d$, what we need to show is the conditional Lindeberg condition: for any nonzero vector $\lambda\in \mathbb{R}^d$,
\begin{equation*}
\sum_{t=1}^n\mathbb{E}[ \rho'(e_t)^2(\lambda^\top M_n^{-1/2}\xi_{nt})^2I(|\rho'(e_t)\lambda^\top M_n^{-1/2}\xi_{nt}|>\varepsilon)|\ra_{t-1}]\to_P0,
\end{equation*}
for any $\varepsilon>0$. This is true because $\max_t\mathbb{E}[ \rho'(e_t)^4|\ra_{t-1}]<\infty$, and $M_n$ converges to a positive definite matrix so that
\begin{align*}
 &\sum_{t=1}^n(\lambda^\top M_n^{-1/2}\xi_{nt})^4\le C\sum_{t=1}^n\|\lambda\|^4\|\xi_{nt}\|^4\\
 \le&C\frac{1}{n}\sum_{t=1}^n\dot{g}_1^4(x_{1t})x_{1t}^4+
 C\frac{1}{n^3}\sum_{t=1}^n\dot{g}_1(x_{1t})^4\|x_{2t}\|^4+C
 \frac{1}{n^2}\sum_{t=1}^n\dot{g}_2^4( z_t^\top\beta_0) \|z_t\|^4  \overset{P}{\to}\, 0
\end{align*}
due to \citet{phillips2000} and the integrability of $u^4\dot{g}^4_1(u)$ for the first two terms and the LLN for the last one.
\end{proof}

\begin{lemma}\label{lemma5}
Suppose that $\rho'(u)$ and/or $\rho''(u)$ exists in generalized sense. Under Assumptions A-C and D(2)-(3), as $n\to\infty$ we have jointly for each $m$ large,
\begin{align}\label{lemma5aa}
 \sum_{t=1}^n\rho''_m(e_t)M_{n,t}\to_D a_{2m}\,M,
\end{align}
where $a_{2m}=\mathbb{E}\rho''_m(e_t)>0$ and $M$ is the same as in Lemma \ref{lemma4}.

Meanwhile,
\begin{align}\label{lemma5c}
 \sum_{t=1}^n [\rho'_m(e_t)-\mathbb{E}\rho'_m(e_t)]\xi_{nt}
\to_D M^{1/2}\xi \ \ \mbox{with} \ \ \xi\sim N(0, a_{1m}I_d),
\end{align}
where $a_{1m}=\text{Var}[\rho'_m(e_t)]$.
\end{lemma}

\begin{proof}[Proof of Lemma \ref{lemma5}]
The proof is the same as that of Lemma \ref{lemma4}, so is omitted.
\end{proof}

\section*{Appendix C: Proofs of the main results}

{\bf Proof of Theorem \ref{th1}}\ \ Letting $p_n, q_n\to\infty$ with $n$ (determined later) and $\alpha=p_n(\theta-\theta_0)$, $\gamma=q_n(\beta-\beta_0)$, we write
\begin{align*}
L_n(\theta,\beta)=&\sum_{t=1}^n [\rho(y_t-g(x_{t}^\top\theta, z_t^\top\beta))-\rho(y_t-g(x_{t}^\top\theta_0, z_t^\top\beta_0))]\\
=&\sum_{t=1}^n [\rho(e_t+g(x_{t}^\top\theta_0, z_t^\top\beta_0)- g(x_{t}^\top\theta, z_t^\top\beta))-\rho(e_t)]\\
=&\sum_{t=1}^n [\rho(e_t+g(x_{t}^\top\theta_0, z_t^\top\beta_0) -g(x_{t}^\top\theta_0+x_t^\top p_n^{-1}\alpha, z_t^\top\beta_0+z_t^\top q_n^{-1}\gamma))-\rho(e_t)]\\
\equiv&Q_n(\alpha, \gamma).
\end{align*}
Hence, if $(\widehat\alpha,\widehat\gamma)$ minimizes $Q_n(\alpha, \gamma)$, $(\widehat\theta,\widehat\beta)$ minimizes $L_n(\theta, \beta)$ where $\widehat\alpha=p_n(\widehat\theta-\theta_0) $ and $\widehat\gamma= q_n(\widehat\beta -\beta_0)$. This reparameterization is adopted in the literature quite often, for example, \citet{bickel1974}, \citet{badu1989}, \citet{davis1992} and \citet{phillips1995}.

For simplicity we denote $\Delta_t(\alpha,\gamma)= g(x_{t}^\top\theta_0, z_t^\top\beta_0) -g(x_{t}^\top\theta_0+x_t^\top p_n^{-1}\alpha, z_t^\top\beta_0+z_t^\top q_n^{-1}\gamma)$, and we then have $Q_n(\alpha, \gamma)=\sum_{t=1}^n [\rho(e_t+\Delta_t(\alpha,\gamma)) -\rho(e_t)]$. Certainly, the $Q_n(\alpha, \gamma)$ only serves for theoretical analysis.

Using Taylor expansion, we have
\begin{align*}
Q_n(\alpha, \gamma)=&\sum_{t=1}^n [\rho(e_t+\Delta_t(\alpha,\gamma)) -\rho(e_t)] = \sum_{t=1}^n\left[ \rho'(e_t)\Delta_t(\alpha,\gamma) +\frac{1}{2}\rho''(e_t^*)\Delta_t(\alpha,\gamma)^2\right],
\end{align*}
where $e_t^*$ is between $e_t$ and $e_t+\Delta_t(\alpha,\gamma)$. This raises an issue: for some loss function, $\rho'(\cdot)$ and $\rho''(\cdot)$ may not exist everywhere. Although it is possible to have Taylor expansion for nonsmooth $\rho$ (see \citet{estrada1993}) by regarding it as a generalized function, in order to facilitate the development of an asymptotic theory in the sequel, we make a detour to invoke the regular sequence
\begin{align*}
  \rho_m(u)=&\int \rho(x)\phi_m(u-x)dx,  \ \   m=1,2,\cdots ,
\end{align*}
where $\phi_m(x)\in S$ is a delta-convergent sequence that ensures $\rho_m(u)\to \rho(u)$ as $m\to\infty$. There are many choices of $\phi_m(x)$ such as $\phi_m(x)=\sqrt{m}\exp(-mx^2)$.

It follows that
\begin{align*}
  \rho_m'(u)=&-\int \rho(x)\phi_m'(x-u)dx, & \rho_m''(u)=&\int \rho(x)\phi_m''(x-u)dx, \ \   m=1,2,\cdots ,
\end{align*}
and $\rho'_m(u)$ and $\rho''_m(u)$ will converge to $\rho'(u)$ and $\rho''(u)$, respectively, in generalized function sense. Accordingly, $\mathbb{E}[\rho'_m(e_1)]\to \mathbb{E}[\rho'(e_1)]=0$ as $m\to\infty$.

\bigskip

Hence, whenever $\rho'(\cdot)$ and/or $\rho''(\cdot)$ may not exist point-wisely, instead of $Q_n(\alpha,\gamma)$ we consider
\begin{equation*}
Q_{mn}(\alpha,\gamma)=\sum_{t=1}^n [\rho'_m(e_t)-\mathbb{E}\rho'_m(e_t)] \Delta_t(\alpha,\gamma)+\frac{1}{2}\rho''_m(e_t^*)
\Delta_t(\alpha,\gamma)^2,
\end{equation*}
where the term $\rho'_m(e_t)-\mathbb{E}\rho'_m(e_t)$ not only converges to  $\rho'(e_t)$ as $m\to\infty$, but also keeps the zero mean condition just as $\rho'(e_t)$.

By Assumption D(1), $g(u,v)=g_1(u)+g_2(v)$ where $g_1(u)$, $\dot{g}_1(u)$ and $\ddot{g}_1(u)$ are $H$-regular with homogeneity order $\nu(\cdot)$, $\dot{\nu}(\cdot)$ and $\ddot{\nu}(\cdot)$, respectively, such that $\lim_{\lambda\to +\infty} \dot{\nu}(\lambda)/\nu(\lambda)=0$, and $\lim_{\lambda\to +\infty} \ddot{\nu}(\lambda)/\dot{\nu}(\lambda)=0$. In this case, we take $p_n=\dot\nu(\sqrt{n})\, n$ and $q_n=\sqrt{n}$.

Further, for any $(\alpha,\gamma)\in \mathbb{R}^{d_1+d_2}$,
\begin{align*}
\Delta_t(\alpha,\gamma)=&g(x_{t}^\top\theta_0, z_t^\top\beta_0) -g(x_{t}^\top\theta_0+x_t^\top p_n^{-1}\alpha, z_t^\top\beta_0+z_t^\top q_n^{-1}\gamma)\\
=&-\dot{g}_1(x_{t}^\top\theta_0) x_t^\top p_n^{-1}\alpha -\dot{g}_2( z_t^\top\beta_0) z_t^\top q_n^{-1}\gamma \\
&-\frac{1}{2}\ddot{g}_1(x_{t}^\top\theta^*) (x_t^\top p_n^{-1}\alpha)^2 -\frac{1}{2}\ddot{g}_2(z_t^\top\beta^*) (z_t^\top q_n^{-1}\gamma)^2 \\
=&-\dot{g}_1(x_{t}^\top\theta_0) x_t^\top p_n^{-1}\alpha -\dot{g}_2( z_t^\top\beta_0) z_t^\top q_n^{-1}\gamma +\text{reminder term},
\end{align*}
where the reminder term is of smaller order than the first two. Indeed, observe that
\begin{align*}
 &\ddot{g}_1(x_{t}^\top\theta^*) (x_t^\top p_n^{-1}\alpha)^2
 =\ddot{g}_1(x_{nt}^\top\theta^*) (x_{nt}^\top \alpha)^2\frac{\ddot{\nu}(\sqrt{n})}{[\dot{\nu}(\sqrt{n})]^2}\frac{1}{n}=O_P
 \left(\frac{\ddot{\nu}(\sqrt{n})}{[\dot{\nu}(\sqrt{n})]^2}\frac{1}{n}\right),\\
 &\ddot{g}_2(z_t^\top\beta^*) (z_t^\top q_n^{-1}\gamma)^2= \ddot{g}_2( z_t^\top\beta^*) (z_t^\top \gamma)^2\frac{1}{n}=O_P(n^{-1}),
\end{align*}
which by Assumption D(1) implies the negligibility of the reminder.

We then write $\Delta_t(\alpha,\gamma)=-\dot{g}_1(x_{t}^\top\theta_0) x_t^\top p_n^{-1}\alpha -\dot{g}_2(z_t^\top\beta_0)z_t^\top q_n^{-1}\gamma$ ignoring the higher order terms.

Moreover, we shall also show
\begin{align*}
Q_{mn}(\alpha, \gamma)
=&\sum_{t=1}^n[\rho'_m(e_t)-\mathbb{E}\rho'_m(e_t)] \Delta_t(\alpha,\gamma)
+\frac{1}{2}\sum_{t=1}^n\rho''_m(e_t^*)\Delta_t(\alpha,\gamma)^2\\
=&\sum_{t=1}^n[\rho'_m(e_t)-\mathbb{E}\rho'_m(e_t)] \Delta_t(\alpha,\gamma)
+\frac{1}{2}\sum_{t=1}^n\rho''_m(e_t)\Delta_n(\alpha,\gamma)^2+r_{mn},
\end{align*}
where we define and will show $r_{mn}= \frac{1}{2}\sum_{t=1}^n [\rho''_m(e_t) -\rho''_m(e_t^*)] \Delta_n(\alpha,\gamma)^2=o_P(1)$.

Notice that for any integer $m$, $\rho''_m(\cdot)$ satisfies Lipschitz condition with the Lipschitz constant $K_m$ that may depend on $m$, because $\phi_m\in S$. Hence,
\begin{align*}
|r_{mn}|
\le&\sum_{t=1}^n|\rho''_m(e_t)-\rho''_m(e_t^*)|\Delta_n(\alpha,\gamma)^2\\
\le&\sum_{t=1}^nK_m |e_t-e_t^*| \Delta_n(\alpha,\gamma)^2 \le K_m\sum_{t=1}^n|\Delta_n(\alpha, \gamma)|^3.
\end{align*}

Further, recalling $p_n$, $q_n$ and $x_{nt}$ we have
\begin{align*}
|r_{mn}|\le& K_m\sum_{t=1}^n|\Delta_n(\alpha, \gamma)|^3\\
=&K_m\sum_{t=1}^n|\dot{g}_1(x_{t}^\top\theta_0) x_t^\top p_n^{-1}\alpha +\dot{g}_2(z_t^\top\beta_0)z_t^\top q_n^{-1}\gamma|^3\\
=&K_mn^{-3/2}\sum_{t=1}^n|\dot{g}_1(x_{nt}^\top\theta_0) x_{nt}^\top \alpha +\dot{g}_2(z_t^\top\beta_0)z_t^\top \gamma|^3 = K_mO_P(n^{-1/2})\to_P 0, \ \ \ \ \text{as}\; n\to\infty,
\end{align*}
for any $m$ and uniformly on any compact set of $(\alpha, \gamma)$.

Rewrite the leading term of $Q_{mn}(\alpha, \gamma)$,
\begin{align*}
Q_{mn}(\alpha, \gamma)
=&\sum_{t=1}^n[\rho'_m(e_t)-\mathbb{E}\rho'_m(e_t)] \Delta_t(\alpha,\gamma)
+\frac{1}{2}\sum_{t=1}^n\rho''_m(e_t)\Delta_t(\alpha,\gamma)^2\\
=&-\sum_{t=1}^n[\rho'_m(e_t)-\mathbb{E}\rho'_m(e_t)]\dot{g}_1(x_{t}^\top\theta_0 ) x_t^\top p_n^{-1}\alpha \\ &-\sum_{t=1}^n[\rho'_m(e_t)-\mathbb{E}\rho'_m(e_t)]\dot{g}_2( z_t^\top\beta_0) z_t^\top q_n^{-1}\gamma\\
&+\frac{1}{2}\sum_{t=1}^n\rho''_m(e_t)[\dot{g}_1(x_{t}^\top\theta_0) x_t^\top p_n^{-1}\alpha]^2+\frac{1}{2}\sum_{t=1}^n\rho''_m(e_t)
[\dot{g}_2(z_t^\top\beta_0)z_t^\top q_n^{-1}\gamma]^2\\
&+\sum_{t=1}^n\rho''_m(e_t)[\dot{g}_1(x_{t}^\top\theta_0) x_t^\top p_n^{-1}\alpha][\dot{g}_2(z_t^\top\beta_0)z_t^\top q_n^{-1}\gamma]\\
=&-\frac{1}{\sqrt{n}}\sum_{t=1}^n[\rho'_m(e_t)-\mathbb{E}\rho'_m(e_t)]
\dot{g}_1(x_{nt}^\top\theta_0) x_{nt}^\top \alpha\\
&-\frac{1}{\sqrt{n}}\sum_{t=1}^n [\rho'_m(e_t)-\mathbb{E}\rho'_m(e_t)] \dot{g}_2(z_t^\top\beta_0)z_t^\top \gamma\\
&+\frac{1}{2n}\sum_{t=1}^n\rho''_m(e_t)[\dot{g}_1(x_{nt}^\top\theta_0) x_{nt}^\top \alpha]^2+ \frac{1}{2n} \sum_{t=1}^n \rho''_m(e_t) [\dot{g}_2( z_t^\top\beta_0)z_t^\top \gamma]^2\\
&+\frac{1}{n}\sum_{t=1}^n\rho''_m(e_t)[\dot{g}_{1}(x_{nt}^\top\theta_0) x_{nt}^\top \alpha][\dot{g}_2(z_t^\top\beta_0)z_t^\top \gamma],
\end{align*}
which is quadratic and hence convex in $(\alpha, \gamma)$ by noting the non-negativeness of $\rho_m''(\cdot)$.

Using the joint convergence in Lemma \ref{lemma3}, for each $m$, as $n\to \infty$, $Q_{mn}$ converges in distribution to
\begin{align*}
&-\int_0^1\dot{g}_1(B(r)^\top\theta_0)B(r)^\top dU_{m}(r)\, \alpha -\xi_m^\top \gamma \\
&+\frac{1}{2}\alpha^\top\left(a_{m2}\int_0^1\dot{g}_1(B(r)^\top\theta_0)
B(r)B(r)^\top dr\right)\alpha +\frac{1}{2}\gamma^\top\left(a_{m2}\Sigma\right)\gamma\\
&+\gamma^\top\left(a_{m2}\int_0^1 \dot{g}_1(B(r)^\top\theta_0) [\mathbb{E}\dot{g}_2(h(r,v_1)^\top\beta_0) h(r,v_1)]B(r)^\top dr\right)\alpha
\end{align*}
where $\xi_m$ is a normal vector variable with covariance $a_{m1}\Sigma$ given in Lemma \ref{lemma3}. Notice further that $a_{m1}\to a_1=\mathbb{E}[\rho'(e_t)^2|\ra_{t-1}]$ and $a_{m2}\to a_2=\mathbb{E}[\rho''(e_t)|\ra_{t-1}]$ which implies $\xi_m\to_Pb_2\sim N(0, a_1\Sigma)$ and $\sup_{r\in [0,1]} \mathbb{E}[U_{m}(r)- U(r)]^2\to 0$ as $m\to\infty$; then $Q_{mn}$ has sequential limit in distribution
\begin{align}\label{limit}
Q(\alpha, \gamma)=&-b_1^\top \alpha -b_2^\top \gamma
+\frac{1}{2}\alpha^\top A_1\alpha +\frac{1}{2}\gamma^\top A_2\gamma
+\gamma^\top A_3\alpha
\end{align}
when $n\to\infty$ firstly and $m\to\infty$ secondly, where
\begin{align*}
  b_1= & \int_0^1\dot{g}_1(B(r)^\top\theta_0)B(r) dU(r), &
  b_2\sim&N(0, a_1\Sigma),\\
  A_1= & a_{2}\int_0^1\dot{g}_1(B(r)^\top\theta_0)B(r)B(r)^\top dr, &
  A_2=&a_2\Sigma,\\
  A_3=&a_{2}\int_0^1 \dot{g}_1(B(r)^\top\theta_0) [\mathbb{E}\dot{g}_2(h(r,v_1)^\top\beta_0) h(r,v_1)]B(r)^\top dr. & &
\end{align*}

The convergence holds uniformly on any compact set of $(\alpha, \gamma)$.

By virtue of the convexity in $(\alpha, \gamma)$ for $Q_{mn}(\alpha, \gamma)$ and $Q(\alpha, \gamma)$, the estimator $(\widehat\alpha, \widehat \gamma)$ will converge to the unique minimizer $(\alpha_0, \gamma_0)$ of the limit $Q(\alpha, \gamma)$ in \eqref{limit} (see Lemma A of \citet{knight1989}), that is,
\begin{align*}
  \begin{pmatrix}
  \dot{\nu}(\sqrt{n})n(\widehat \theta-\theta_0)\\
  \sqrt{n}(\widehat \beta-\beta_0)
  \end{pmatrix}
  \to_D\begin{pmatrix}
  [A_1-A_3^\top A_2A_3]^{-1}(b_1-A^\top_3 A_2^{-1}b_2)\\
  A_2^{-1}(b_2-A_3[A_1-A_3^\top A_2A_3]^{-1}(b_1-A_3^\top A_2^{-1}b_2))
  \end{pmatrix}.
\end{align*}
This finishes the proof. \ \ $\Box$

{\bf Proof of Corollary \ref{cor1}}\ \ When $\rho'$ or $\rho''$ does not exist in ordinary sense, their regular sequences $\rho'_m$ or $\rho''_m$ may be used to replace them. The replacement does not affect the proof as can be seen from the proof of the theorem. Thus, in what follows we only consider the case where $\rho'$ and $\rho''$ exist and are continuous.

By the definition,
\begin{align*}
\widehat{e}_t=&y_t-g(\widehat{\theta}^{\,\top} x_t, \widehat{\beta}^{\,\top} z_t)=e_t+g({\theta}^{\,\top}_0 x_t, {\beta}^{\,\top}_0 z_t) -g(\widehat{\theta}^{\,\top} x_t, \widehat{\beta}^{\,\top} z_t) = e_t+\Delta_t(\widehat{\theta}, \widehat{\beta}),
\end{align*}
where $\Delta(\theta,\beta)=g({\theta}^{\,\top}_0 x_t, {\beta}^{\,\top}_0 z_t)-g({\theta}^{\,\top} x_t, {\beta}^{\,\top} z_t)$ and it is shown in the proof of the theorem $\Delta_t(\widehat{\theta}, \widehat{\beta}) =-\dot{g}_1({\theta}^{\,\top}_0 x_t)x_t^\top(\widehat{\theta} -\theta_0)-\dot{g}_2({\beta}^{\,\top}_0 z_t)z_t^\top(\widehat{\beta} -\beta_0)+$higher order term.

Thus, we simply take $\Delta_t(\widehat{\theta}, \widehat{\beta}) =-\dot{g}_1({\theta}^{\,\top}_0 x_t)x_t^\top(\widehat{\theta} -\theta_0)-\dot{g}_2({\beta}^{\,\top}_0 z_t) z_t^\top(\widehat{\beta}-\beta_0)$ and further by virtue of Theorem \ref{th1}, $\Delta_t(\widehat{\theta}, \widehat{\beta}) =O_P(n^{-1/2})$ uniformly in $t$. It then follows from the mean value theorem that
\begin{align*}
 &[\rho'(\widehat{e}_t)]^2-[\rho'(e_t)]^2=[\rho'(e_t+\Delta_t(\widehat{\theta}, \widehat{\beta}))]^2 -[\rho'(e_t)]^2=2\rho'(e_t^*) \rho''(e_t^*) \Delta_t(\widehat{\theta}, \widehat{\beta})
\end{align*}
where $e_t^*$ is on the segment joining $e_t$ and $\widehat{e}_t$. Hence,
\begin{align*}
\widehat{a}_1-a_1=&\frac{1}{n}\sum_{t=1}^n \{[\rho'(\widehat{e}_t)]^2-\e [\rho'(e_t)]^2\}\\
=&2 \frac{1}{n}\sum_{t=1}^n\rho'(e_t^*) \rho''(e_t^*) \Delta_t(\widehat{\theta}, \widehat{\beta})+\frac{1}{n}\sum_{t=1}^n \{[\rho'(e_t)]^2-\e [\rho'(e_t)]^2\}.
\end{align*}

The first term is $o_P(1)$ since the continuity of the two derivatives and $|e_t^*-e_t|\le |e_t-\widehat{e}_t|\le \Delta_t(\widehat{\theta}, \widehat{\beta})$ that is uniformly $o_P(1)$. The last term is $o_P(1)$ because of the mean ergodicity of the sequence $\{\rho'(e_t)^2\}_{1}^n$ and $a_1=\e([\rho'(e_t)]^2|\ra_{t-1})<\infty$ by Assumption C. Then the first assertion follows.

Notice further that
\begin{align*}
\widehat{a}_2-a_2=&\frac{1}{n}\sum_{t=1}^n \{\rho''(\widehat{e}_t)-\e [\rho''(e_t)]\}\\
=& \frac{1}{n}\sum_{t=1}^n[\rho''(\widehat{e}_t)-\rho''(e_t)] +\frac{1}{n}\sum_{t=1}^n \{\rho''(e_t)-\e [\rho''(e_t)]\}.
\end{align*}
Here, the first term is $o_P(1)$ by the continuity of the $\rho''$ and the uniform approximation of $|e_t-\widehat{e}_t|$; the second term is $o_P(1)$ because $\{\rho''(e_t)\}_1^n$ is ergodic. \ $\Box$

\medskip

{\bf Proof of Theorem \ref{th2}}\ \ Note that $\|\theta_0\|\ne 0$ and we allow that $\theta_0$ may not be a unit vector, so the treatment in what follows is different from the existing literature on semiparametric single index models. We take $P=(\theta_0/\|\theta_0\|, P_2)_{d_1\times d_1}$ to be an orthogonal matrix that will be used to rotate the regressor $x_t$. Define $D_n=\text{diag}(\sqrt[4]{n}, \sqrt[4]{n}^3I_{d_1-1})$ and $q_n=\sqrt{n}$ for later use.

Denoting $\Delta_t(\theta, \beta)=g(x_t^\top \theta_0, z_t^\top \beta_0)-g(x_t^\top \theta, z_t^\top \beta)$, by Taylor theorem we have
\begin{align*}
L_n(\theta, \beta)=&\sum_{t=1}^n[\rho(y_t-g(x_t^\top \theta, z_t^\top \beta))-\rho(y_t-g(x_t^\top \theta_0, z_t^\top \beta_0))]\\
=&\sum_{t=1}^n[\rho(e_t+\Delta_t(\theta, \beta))-\rho(e_t)] = \sum_{t=1}^n[\rho'(e_t)\Delta_t(\theta, \beta)+\frac{1}{2} \rho''(e_t^*) \Delta_t(\theta, \beta)^2],
\end{align*}
where $e_t^*$ is somewhere between $e_t$ and $e_t+\Delta_t(\theta, \beta)$, and we contemporarily consider the case that the first and second derivatives of $\rho(\cdot)$ exist while otherwise the cases will be dealt with at a later stage. Moreover, using the similar idea as in the proof of Theorem 2.1, the $\rho''(e_t^*)$ in the second term above can be approximated by $\rho''(e_t)$ because as shown in Theorem 2.1
\begin{align*}
 \sum_{t=1}^n|\rho''(e_t^*)-\rho''(e_t)|\Delta_t(\theta, \beta)^2\le C \sum_{t=1}^n|\Delta_t(\theta, \beta)|^3,
\end{align*}
and using first order Taylor expansion and Assumption D(3),
\begin{align*}
\Delta_t(\theta, \beta)=&g(x_t^\top \theta_0, z_t^\top \beta_0)-g(x_t^\top \theta, z_t^\top \beta)\\
=&-\dot{g}_1(x_t^\top \theta^*)x^\top_t (\theta-\theta_0)-\dot{g}_2(z_t^\top \beta^*)z_t^\top (\beta-\beta_0),
\end{align*}
where $(\theta^*,\beta^*)$ is some point on the segment joining $(\theta,\beta)$ and $(\theta_0,\beta_0)$. Here we shall use $\theta^*$ to construct a matrix $P^*=(\theta^*, P_2^*)$ such that $\theta^*$ is orthogonal to each row of $P_2^*$, similar to the construction of $P$. Then, $x^\top_t (\theta-\theta_0) =x_t^\top P^*D_n^{-1}D_n{P^*}' (\theta-\theta_0)=(x_{1t}^*, (x_{2t}^*)^\top)^\top D_n^{-1}\widetilde{\alpha}$, where $x_{1t}^*= x_t^\top \theta^*$, $x_{2t}^*=(P_2^*)^\top x_t$ and $\widetilde{\alpha}=D_n{P^*}' (\theta-\theta_0)$ that is a scaled vector. Also, rewrite $z_t^\top (\beta-\beta_0)=z_t^\top q_n^{-1}\gamma$ with $\gamma=q_n(\beta- \beta_0)$ a scaled vector too. We then may regard the scaled parameters $\widetilde{\alpha}$ and $\gamma$ as generic parameters in the minimization we consider later.

Moreover, we bound $\Delta_t(\theta, \beta)$ by three additive components, ensuing that
\begin{align*}
\sum_{t=1}^n|\Delta_t(\theta, \beta)|^3\le& C_1\frac{1}{n^{3/4}} \sum_{t=1}^n |\dot{g}_1(x_{1t}^*)x_{1t}^*|^3+C_2\frac{1}{n^{9/4}} \sum_{t=1}^n |\dot{g}_1(x_{1t}^*)|^3\|x_{2t}^*\|^3\\
&+C_3\frac{1}{n^{3/2}} \sum_{t=1}^n\|\dot{g}_2(z_t^\top \beta^*)z_t\|^3 = O_P(n^{-1/2}),
\end{align*}
evaluated by using Lemma \ref{lemma1} for the first two terms and LLN for the last term. Now, we can write
\begin{align*}
L_n(\theta, \beta)=&\sum_{t=1}^n[\rho'(e_t)\Delta_t(\theta, \beta)+\frac{1}{2}\rho''(e_t) \Delta_t(\theta, \beta)^2]+O_P(n^{-1/2}).
\end{align*}

We further consider the approximation of $\Delta_t(\theta, \beta)$. By D(3), again Taylor theorem expanding $g(u,v)$ up to the second order and the orthogonal $P$ and a new matrix $P^*$ that is constructed as before but most unlikely the same,
\begin{align*}
\Delta_t(\theta, \beta)=&-\dot{g}_1(x_{t}^\top\theta_0) x_t^\top P D_n^{-1}D_nP^\top(\theta-\theta_0) -\dot{g}_2( z_t^\top\beta_0) z_t^\top q_n^{-1}q_n(\beta-\beta_0) \\
&-\frac{1}{2}\ddot{g}_1(x_{t}^\top\theta^*) (x_t^\top P^* D_n^{-1}D_nP^{*\top}(\theta-\theta_0))^2 -\frac{1}{2}\ddot{g}_2(z_t^\top\beta^*) (z_t^\top q_n^{-1}q_n(\beta-\beta_0))^2\\
=&-\dot{g}_1(x_{t}^\top\theta_0) x_t^\top P D_n^{-1}\widetilde\alpha -\dot{g}_2( z_t^\top\beta_0) z_t^\top q_n^{-1}\gamma \\
&-\frac{1}{2}\ddot{g}_1(x_{t}^\top\theta^*) (x_t^\top P^* D_n^{-1}{\widetilde\alpha}^*)^2 -\frac{1}{2}\ddot{g}_2(z_t^\top\beta^*) (z_t^\top q_n^{-1}\gamma)^2,
\end{align*}
where $\widetilde\alpha=D_nP^\top (\theta-\theta_0)$, $\gamma=q_n(\beta- \beta_0)$ and ${\widetilde\alpha}^*=D_nP^{*\top} (\theta-\theta_0)$. We then may follow the same derivation as above, along with the martingale difference structure for $(\rho'(e_t),\ra_t)$, to show that the second order term in the expansion is negligible, that is,
\begin{align*}
 &\sum_{t=1}^n\rho'(e_t)\ddot{g}_1(x_{t}^\top\theta^*) (x_t^\top P^* D_n^{-1}{\widetilde\alpha}^*)^2=o_P(1),\ \
 \sum_{t=1}^n\rho'(e_t)\ddot{g}_2(z_t^\top\beta^*) (z_t^\top q_n^{-1}\gamma)^2=o_P(1),\\
  &\sum_{t=1}^n\rho''(e_t)\ddot{g}_1(x_{t}^\top\theta^*) (x_t^\top P^* D_n^{-1}{\widetilde\alpha}^*)^4=o_P(1),\ \
 \sum_{t=1}^n\rho''(e_t)\ddot{g}_2(z_t^\top\beta^*) (z_t^\top q_n^{-1}\gamma)^4=o_P(1).
\end{align*}
Thus, this higher order term is omitted in the sequel.

Notice further that, denoting $\widetilde\alpha=(\widetilde\alpha_1, \widetilde\alpha_2^\top)^\top$ conformably with $x_{1t}=x_t^\top \theta_0$ and $x_{2t}^\top=x_t^\top P_2$ (hereafter denote ${\widetilde x}_t=(x_{1t}, x_{2t}^\top)^\top$),
\begin{align*}
 x_t^\top P D_n^{-1}\widetilde\alpha=&(x_t^\top \theta_0/\|\theta_0\|, x_t^\top P_2)D_n^{-1}\widetilde\alpha = {\widetilde x}_t^\top D_n^{-1}(\widetilde\alpha_1/\|\theta_0\|, \widetilde\alpha_2^\top)^\top\equiv {\widetilde x}_t^\top D_n^{-1}\alpha,
\end{align*}
where we denote
\begin{align}\label{alphatheta}
 \alpha=D_n\begin{pmatrix} \frac{1}{\|\theta_0\|^2} \theta_0^\top\\
 P_2^\top
 \end{pmatrix}(\theta-\theta_0).
\end{align}

The transformation from $x_t$ to ${\widetilde x}_t$ is necessary since $g_1(\cdot)$ is $I$-regular. See \citet{dgd2016} and the reference therein for more details.

Accordingly, $L_n(\theta, \beta)$ is reparameterized as $Q_n(\alpha, \gamma)$
\begin{align*}
Q_n(\alpha, \gamma)=&\sum_{t=1}^n\big\{\rho'(e_t)[-\dot{g}_1(x_{t}^\top\theta_0) {\widetilde x}_t D_n^{-1}\alpha -\dot{g}_2( z_t^\top\beta_0) z_t^\top q_n^{-1}\gamma]\\
&+\frac{1}{2}\rho''(e_t)[\dot{g}_1(x_{t}^\top\theta_0) {\widetilde x}_t D_n^{-1}\alpha -\dot{g}_2( z_t^\top\beta_0) z_t^\top q_n^{-1}\gamma]^2\big\}.
\end{align*}

Note that, when $(\widehat\alpha, \widehat\gamma)$ minimizes $Q_n(\alpha, \gamma)$, $(\widehat\theta, \widehat\beta)$ minimizes $L_n(\theta, \beta)$ where $\widehat\gamma=q_n(\widehat\beta- \beta_0)$ and $\widehat\theta$ is given through the relationship \eqref{alphatheta} with $\alpha$ and $\theta$ being respectively replaced by $\widehat\alpha$ and $\widehat\theta$.

When $\rho'(\cdot)$ or/and $\rho''(\cdot)$ may not exist at some points, similarly to Theorem 2.1 we consider
\begin{align*}
Q_{mn}(\alpha, \gamma)=&\sum_{t=1}^n\big\{[\rho'_m(e_t)- \mathbb{E}\rho'_m(e_t)] [-\dot{g}_1(x_{t}^\top\theta_0) {\widetilde x}_t D_n^{-1}\alpha -\dot{g}_2( z_t^\top\beta_0) z_t^\top q_n^{-1}\gamma]\\
&+\frac{1}{2}\rho''_m(e_t)[\dot{g}_1(x_{t}^\top\theta_0) {\widetilde x}_t D_n^{-1}\alpha -\dot{g}_2( z_t^\top\beta_0) z_t^\top q_n^{-1}\gamma]^2\big\},
\end{align*}
where $\rho'_m(\cdot)$ or/and $\rho''_m(\cdot)$ are regular sequences of $\rho'(\cdot)$ and $\rho''(\cdot)$.

Hence,
\begin{align*}
&Q_{mn}(\alpha,\gamma)
=-\sum_{t=1}^n [\rho'_m(e_t)-\mathbb{E}\rho'_m(e_t)] [\dot{g}_1(x_{t}^\top\theta_0) {\widetilde x}_t^\top D_n^{-1}\alpha -\dot{g}_2( z_t^\top\beta_0) z_t^\top q_n^{-1}\gamma]\\
&\hspace{2.3cm}+\frac{1}{2}\sum_{t=1}^n\rho''_m(e_t)[\dot{g}_1(x_{t}^\top\theta_0) {\widetilde x}_t^\top D_n^{-1}\alpha -\dot{g}_2( z_t^\top\beta_0) z_t^\top q_n^{-1}\gamma]^2\\
=&-(\alpha_1, \alpha_2^\top, \gamma^\top) \sum_{t=1}^n [\rho'_m(e_t)- \mathbb{E}\rho'_m(e_t)] \xi_{nt}
+\frac{1}{2}(\alpha_1, \alpha_2^\top, \gamma^\top)\sum_{t=1}^n\rho''_m(e_t)
 M_{nt}\begin{pmatrix} \alpha_1\\ \alpha_2 \\ \gamma \end{pmatrix}\\
\end{align*}
where we define $\xi_{nt}:= \begin{pmatrix} \frac{1}{\sqrt[4]{n}}\dot{g}_1(x_{1t})x_{1t}\\
 \frac{1}{\sqrt[4]{n}^3} \dot{g}_1(x_{1t})x_{2t}\\
 \frac{1}{\sqrt{n}}\dot{g}_2( z_t^\top\beta_0) z_t
\end{pmatrix}$ and $M_{nt}:=\xi_{nt}\xi_{nt}^\top$.

It follows from the joint convergence in Lemmas \ref{lemma4}-\ref{lemma5} that, as $n\to\infty$ and $m\to\infty$,
$$Q_{mn}(\alpha, \gamma) \to_DQ(\alpha,\gamma)=-(\alpha_1, \alpha_2^\top, \gamma^\top)\xi+ \frac{a_2}{2}(\alpha_1, \alpha_2^\top, \gamma^\top)M(\alpha_1, \alpha_2^\top, \gamma^\top)^\top$$
uniformly in any compact subset of $\mathbb{R}^d$ where $\xi\sim N(0, a_1 M)$.

By virtue of the convexity in $(\alpha, \gamma)$ for $Q_{mn}(\alpha, \gamma)$ and $Q(\alpha, \gamma)$, the estimator $(\widehat\alpha, \widehat \gamma)$ will converge to the unique minimizer $(\alpha_0, \gamma_0)$ of the limit $Q(\alpha, \gamma)$ in \eqref{limit} (see Lemma A of \citet{knight1989}), that is,
\begin{align*}
\begin{pmatrix}\widehat\alpha\\ \widehat \gamma \end{pmatrix}=\begin{pmatrix} \sqrt[4]{n}\frac{1}{\|\theta_0\|^2}\theta_0^\top (\widehat\theta-\theta_0)\\
\sqrt[4]{n}^3P_1^\top (\widehat\theta-\theta_0)\\
q_n(\widehat\beta-\beta_0)
\end{pmatrix}\to_D(a_2M)^{-1}\xi=_D\frac{\sqrt{a_1}}{a_2}M^{-1/2}N(0,I_d).
\end{align*}
The proof is finished. \ \ $\Box$

{\bf Proof of Corollary \ref{cor2}}\ \ We now consider the convergence of $\widehat\theta-\theta_0$. Denote $J=\text{diag}(1/\|\theta_0\|, I_{d_1-1})$ for convenience. Observe that
\begin{align*}
&\sqrt[4]{n}(\widehat\theta-\theta_0)=\sqrt[4]{n} (D_nJP)^{-1}(D_nJP)(\widehat\theta-\theta_0)\\
=&\sqrt[4]{n} (D_nJP)^{-1}
\begin{pmatrix} \sqrt[4]{n}\frac{1}{\|\theta_0\|^2}\theta_0^\top (\widehat\theta-\theta_0)\\
\sqrt[4]{n}^3P_1^\top (\widehat\theta-\theta_0)
\end{pmatrix}\\
\to_D&P^\top J^{-1} \text{diag}(1, {\bf 0}_{d_1-1})N(0,M^{-1}_1)
=_DN(0, \tau_{11}\theta_0\theta_0^\top) \quad \Box
\end{align*}

}

\end{document}